\documentclass[11pt,reqno]{amsart}
\pdfoutput=1

\usepackage{amsmath, amsthm, amssymb, mathtools, mathrsfs}
\usepackage{amsaddr}
\usepackage{dsfont}
\usepackage{graphicx}
\usepackage{enumitem}
\usepackage{fullpage}
\usepackage[pdftex]{hyperref}
\usepackage{physics}
\usepackage{cleveref}
\usepackage[bibstyle=numeric, backend=biber, sorting=none, citestyle=numeric-comp, giveninits, url=false, maxbibnames=99]{biblatex} % loads etoolbox
\usepackage[toc,page]{appendix}

\usepackage{tikz}
\usetikzlibrary{calc}
\usetikzlibrary{arrows,backgrounds}
\usetikzlibrary{arrows.meta, bending, positioning}

\newcommand{\tikzmath}[2][]
{\vcenter{\hbox{\begin{tikzpicture}[#1]#2\end{tikzpicture}}}
}

\newcommand{\roundNbox}[6]{
	\draw[rounded corners=5pt, very thick, #1] ($#2+(-#3,-#3)+(-#4,0)$) rectangle ($#2+(#3,#3)+(#5,0)$);
	\coordinate (ZZa) at ($#2+(-#4,0)$);
	\coordinate (ZZb) at ($#2+(#5,0)$);
	\node at ($1/2*(ZZa)+1/2*(ZZb)$) {#6};
}

\usetikzlibrary{decorations,decorations.pathreplacing,decorations.markings}

\tikzset{snake it/.style={decorate, decoration=snake}}
\tikzset{super thick/.style={line width=3pt}}

\tikzstyle{mid>}=[decoration={markings, mark=at position 0.5 with {\arrow{>}}}, postaction={decorate}]
\tikzstyle{mid<}=[decoration={markings, mark=at position 0.5 with {\arrow{<}}}, postaction={decorate}]
\tikzstyle{knot}=[preaction={super thick, white, draw}]
\tikzstyle{gray knot}=[preaction={super thick, gray!25, draw}]

\usepackage{amsmath, amsthm, amssymb, mathtools, mathrsfs}

\usepackage{xcolor}
\definecolor{rufous}{HTML}{A81C07}
\definecolor{azure}{HTML}{007fff}
\definecolor{OliveGreen}{HTML}{6D712E}
\definecolor{boysenberry}{HTML}{873260}
\definecolor{violet}{RGB}{148,0,211}
\definecolor{salmon}{HTML}{ff8c69}
\definecolor{DarkGreen}{RGB}{0,150,0}
\definecolor{NewOliveGreen}{HTML}{7DA12E}

\newcommand{\Ad}{\operatorname{Ad}}
\newcommand{\id}{\operatorname{id}}

\newcommand{\Irr}{\operatorname{Irr}}
\newcommand{\Dr}{\operatorname{Dr}}
\newcommand{\Tube}{\operatorname{Tube}}

\newcommand{\placeholder}{\bullet}

\def\semicolon{;}
\def\applytolist#1{
    \expandafter\def\csname multi#1\endcsname##1{
        \def\multiack{##1}\ifx\multiack\semicolon
            \def\next{\relax}
        \else
            \csname #1\endcsname{##1}
            \def\next{\csname multi#1\endcsname}
        \fi
        \next}
    \csname multi#1\endcsname}

\def\calc#1{\expandafter\def\csname c#1\endcsname{{\mathcal #1}}}
\applytolist{calc}QWERTYUIOPLKJHGFDSAZXCVBNM;
\def\bbc#1{\expandafter\def\csname bb#1\endcsname{{\mathbb #1}}}
\applytolist{bbc}QWERTYUIOPLKJHGFDSAZXCVBNM;
\def\bfc#1{\expandafter\def\csname bf#1\endcsname{{\mathbf #1}}}
\applytolist{bfc}QWERTYUIOPLKJHGFDSAZXCVBNM;
\def\sfc#1{\expandafter\def\csname s#1\endcsname{{\sf #1}}}
\applytolist{sfc}QWERTYUIOPLKJHGFDSAZXCVBNM;
\def\fc#1{\expandafter\def\csname f#1\endcsname{{\mathfrak #1}}}
\applytolist{fc}QWERTYUIOPLKJHGFDSAZXCVBNM;
\def\rmc#1{\expandafter\def\csname rm#1\endcsname{{\mathrm #1}}}
\applytolist{rmc}QWERTYUIOPLKJHGFDSAZXCVBNM;
\def\scrc#1{\expandafter\def\csname scr#1\endcsname{{\mathscr #1}}}
\applytolist{scrc}QWERTYUIOPLKJHGFDSAZXCVBNM;

\newcommand{\End}{\operatorname{End}}
\newcommand{\supp}{\operatorname{supp}}

\newcommand{\diam}{\operatorname{diam}}

\newcommand{\loc}{\operatorname{loc}}

\hypersetup{
    colorlinks=true,                            % false: boxed links; true: colored links
    linkcolor=azure,                          % color of internal links
    citecolor=NewOliveGreen,                          % color of links to bibliography
    filecolor=red,                           % color of file links
    urlcolor=azure                             % color of external links
}

\theoremstyle{plain}
\newtheorem{thm}{Theorem}[section]

\newtheorem{cor}[thm]{Corollary}
\newtheorem{lem}[thm]{Lemma}

\theoremstyle{definition}
\newtheorem{defn}[thm]{Definition}

\newtheorem{rem}[thm]{Remark}

\title{A nonabelian anyon violates Haag duality}
\author{Daniel Wallick, Henrik Wilming}
\address{Leibniz Universit\"{a}t Hannover, Institut f\"{u}r Theoretische Physik, Appelstra\ss{}e 2, 30167 Hannover, Germany}
\date{\today}
\begin{document}
\begin{abstract}
	We show that a superselction sector describing a single nonabelian anyon violates Haag duality, which is equivalent to the fundamental quantum-information theoretic principle of uniqueness of purifications. 
	We provide a simple physics argument as well as a rigorous proof using sector theory, and instantiate our result concretely in Levin--Wen models. 
	We also show that the associated ground state does not allow for quantum steering despite being a pure state.
	Finally, we show that this ground state violates approximate Haag duality, disproving the conjecture that all gapped ground states satisfy this condition. Our result implies that there are gapped phases of matter where Haag duality fails at every point in the phase.
\end{abstract}
\maketitle
\section{Introduction and overview}
Recently, entanglement properties of systems with infinitely many degrees of freedom are receiving more and more attention. Such systems can have entanglement properties that finite systems cannot have \cite{summersMaximalViolationBells1987,keylInfinitelyEntangledStates2002,verch_distillability_2005,keylEntanglementHaagdualityType2006,crannStateConvertibilityNeumann2020,van_luijk_embezzlement_2024,van_luijk_relativistic_2024,vanluijkPureStateEntanglement2025,MR5026085,luijkQuantumSteeringEquivalent2026}, but which can survive in an approximate way in finite systems \cite{vanluijkCriticalFermionsAre2025,vanluijkLargeScaleStructureEntanglement2025,luijkMultipartiteEmbezzlementEntanglement2025}. A notable example of a property that requires infinitely many degrees of freedom (absent further symmetries) is the violation of ``uniqueness of purifications": First consider a bipartite quantum system made up of finitely many degrees of freedom. Divide the system into two parts $A$ and $B$. Then the total Hilbert space is $\cH=\cH_A \otimes \cH_B$.
Uniqueness of purification says that if $\ket\Psi,\ket\Phi\in \cH$ are state vectors such that $\bra \Psi a\otimes \mathds 1 \ket\Psi = \bra\Phi a\otimes \mathds 1 \ket \Phi$ for all operators $a$ acting on $A$, then for every $\varepsilon>0$ there exists a unitary $u$ acting on $B$ such that
\begin{align}
		\left| \bra\Psi \mathds 1\otimes u \ket\Phi \right| \geq 1-\varepsilon.
\end{align}
In words: If two state vectors have the same reduced density matrix on $A$, then they are related by a unitary acting only on $B$. 

In systems with infinitely many degrees of freedom the Hilbert space $\cH$ generally does not factorize, and the subsystems are instead described by commuting von Neumann algebras $M_A$ and $M_B$. 
The von Neumann algebras $M_A$ and $M_B$ contain the local observables in regions $A$ and $B$, respectively.
Uniqueness of purifications then requires that $\bra\Psi a \ket\Psi =\bra \Phi a \ket\Phi$ for all $a \in M_A$ implies that for any $\varepsilon > 0$ there exists a unitary $u\in M_B$ such that $|\bra\Psi u\ket \Phi|\geq 1-\varepsilon$. 

Recently \cite{MR5026085}, it was shown that uniqueness of purifications is equivalent to \emph{Haag duality}, which is the statement
\begin{align}
		M_A = M_B',\qquad M_B' := \{ x\in B(\cH) : [x,y]=0\ \forall y\in M_A \}.
\end{align}
Haag duality plays an important role in (algebraic) quantum field theory \cite{haagLocalQuantumPhysics1996}\footnote{See \cite{RobertsHDFailure} for an early discussion of the failure of Haag duality in the context of gauge symmetries.} and in the rigorous understanding of topologically ordered phases \cite{naaijkensHaagDualityDistal2012, MR3426207, MR4362722, 2509.23734,2605.10693}. 
Recently, it also came into focus in the context of non-invertible symmetries \cite{shaoAdditivityHaagDuality2025a,harlowDisjointAdditivityLocal2025a}. 
It is generally difficult to prove Haag duality for quantum spin systems, but it was recently proven to hold for all Levin--Wen models \cite{2509.23734,2605.10693} (see \cite{naaijkensHaagDualityDistal2012,MR3426207} for prior results in more specialized cases).
It is an open problem whether Haag duality at a single point in a gapped phase of matter already implies Haag duality throughout the phase.

Topologically ordered systems can host so-called \emph{nonabelian anyons}, emergent point-like quasi-particles with nonabelian exchange statistics, which may in principle be used for fault-tolerant quantum computation \cite{wangTopologicalQuantumComputation2010}. 
While anyons can only be created out of the vacuum in pairs of an anyon with its anti-particle, in an infinite systems we may consider moving one of the partners away to infinity. 
Such a situation is then referred to as a \emph{(nonabelian) superselection sector}, because the left-over anyon cannot be destroyed by any local operation. 
This is a consequence of the fact that they can be \emph{detected} by operators acting along arbitrary loops that encompass the anyon: If a local unitary operator $u$ would remove the anyon from state $\ket\Psi$, we could simply choose a loop outside of the support to detect this using an observable $x$. 
But then the two operators would commute, which contradicts the assumption that $x$ can detect the removal of the anyon, because $x$ would take the same expectation value on $\ket\Psi$ and $u\ket\Psi$.

In this work, we show that uniqueness of purifications, hence Haag duality, must always be violated in a nonabelian superselection sector as long as the system is partitioned into two parts, where each part contains an infinitely large connected component that contains balls of arbitrary sizes. We establish our result on three levels. 
First, we provide a very simple and general physics argument. Second, we provide a rather abstract, but rigorous, model-independent result for cone-regions assuming that the vacuum fulfills Haag duality. Finally, we instantiate our result concretely in Levin--Wen models \cite{PhysRevB.71.045110}, which are believed to cover all bosonic, 2+1-dimensional topologically ordered phases admitting a gapped boundary \cite{MR2200691,kitaevModelsGappedBoundaries2012}. 
Our result shows that one cannot prove, even for simply connected regions, that pure ground states of gapped, commuting projector Hamiltonians fulfil Haag duality. 
Moreover, we generalize this result to a weakened form of Haag duality called \emph{approximate Haag duality} \cite{MR4362722}. Haag duality implies approximate Haag duality. 
However, in contrast to Haag duality, it is known that approximate Haag duality is a property of a gapped phase: If it holds at one point in the phase it holds at all points in the phase \cite{MR4362722}. Our result shows that there are gapped phases of matter where approximate Haag duality fails at every point in the phase. Since Haag duality implies approximate Haag duality it follows that Haag duality fails at every point in the gapped phase.

Finally, we also discuss a connection to a recent result showing that the quantum information theoretic primitive of \emph{steering} may survive the violation of Haag duality \cite{luijkQuantumSteeringEquivalent2026}. This is, for example, the case of the surface code when the $A$ subsystem consists of two disjoint infinite cones \cite{naaijkensKosakiLongoIndexClassification2013,fiedlerJonesIndexSecret2017}.
Our results show that the natural ground states in \emph{nonabelian} superselection sectors in general do not allow for steering, which we demonstrate explicitly in Levin--Wen models.  

\subsection{Non-abelian anyons}
We briefly review the required basics of nonabelian anyons, which are localized excitations in a quantum system.
For an in-depth discussion, see for example \cite{wangTopologicalQuantumComputation2010}.
Throughout we consider two-dimensional systems that can host finitely many types of (simple) nonabelian anyons, which we label by roman letters $a,b,c,\ldots$. A special type of anyon is the \emph{vacuum} type, which we label by $1$ instead. To each anyon type $a$ there is a ``dual anyon type" $\bar a$, which we may heuristically think of the anti-particle associated to $a$. The vacuum type is always its own dual, $1=\bar1$.

Anyons are characterized by three operations: fusion, splitting and braiding.
For us, only fusion and splitting will be important. 
Fusion involves bringing two anyons, of types $a$ and $b$, close together, which we denote by writing $a\otimes b$.  
We can try to interpret their composite as an anyon again. In general the resulting state, which depends on the history of the two anyons, will be a superposition of different anyon types. If we measure the anyon type we get a definitive answer $c$ and project the system into the ``fusion outcome" $c$. The outcome probabilities depend on the quantum state. 
In general, there can be multiple \emph{fusion channels} for $a$ and $b$ to fuse into $c$. We denote the number of them by $N^c_{ab}$. Symbollically, we write
\begin{align}
	a\otimes b = \sum_c N^c_{ab} c.
\end{align}
We call anyon type $a$ \emph{abelian} if it \emph{always} fuses with its dual anyon $\bar a$  into the vacuum $1$, i.e. $N_{a\bar a}^c = \delta^{1c}$. 
Otherwise it is called \emph{nonabelian}. 
It is an essential property of topologically ordered system that measuring the anyon type (more precisely, the total topological charge) inside a region can be done by measuring an operator localized along the boundary of the region. Thus we can measure the anyon type of an anyon by a measurement along any loop that encircles the anyon and no other anyon. 

Splitting is an operation that is dual to fusion, where an anyon $c$ is split up into two anyons of types $a$ and $b$.
If an \emph{abelian} anyon $a$ splits into $a$ and $b$, then $b$ must necessarily be the vacuum. 
But for \emph{nonabelian} anyons there are non-trivial $b$ so that $a$ can split into $a$ and $b$. 
This will be the essential feature of nonabelian anyons we will use to show a violation of Haag duality.

\subsection{The physical argument}
\label{sec:PhysicalArgument}
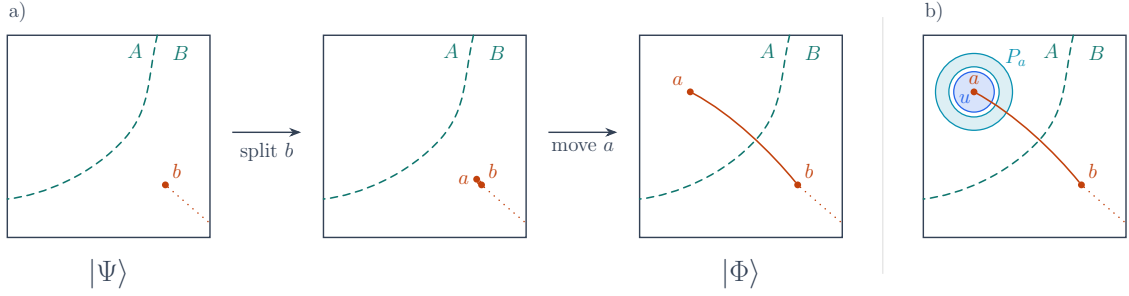
\begin{figure}[t]
  \centering
  % ------------------------------------------------------------------
% anyon_figure_snippet.tex
% Combined figure: a) split-and-move sequence, b) state with regions
% u and P_a around anyon a, separated by a light vertical line.
%
% Usage:
%   \begin{figure}[t]
%     \centering
%     \input{anyon_figure_snippet}
%     \caption{...}
%   \end{figure}
%
% Preamble requirements:
%   \usepackage{tikz}
%   \usetikzlibrary{arrows.meta,calc}
% ------------------------------------------------------------------

% ---------- Color scheme ----------
\definecolor{FrameGray}{HTML}{334155}   % slate: boxes, process arrows, kets
\definecolor{WallTeal}{HTML}{0F766E}    % teal: A/B domain wall + labels
\definecolor{AnyonOrange}{HTML}{C2410C} % burnt orange: anyons, strings, trajectories
\definecolor{DiskBlue}{HTML}{2563EB}    % blue: solid circle u around a
\definecolor{AnnulusCyan}{HTML}{0891B2} % cyan-teal: annulus P_a

% A/B domain wall, identical in every panel; A covers roughly the
% upper-left half of the box.
\providecommand{\domainwall}{}
\renewcommand{\domainwall}{%

  \draw[wall] (3.70,5.00)
      .. controls (3.45,3.90) and (3.55,3.20) .. (2.90,2.45)
      .. controls (2.10,1.55) and (0.90,1.05) .. (0,0.95);
  \node[WallTeal] at (3.15,4.60) {$A$};
  \node[WallTeal] at (4.28,4.55) {$B$};
}

\resizebox{.9\textwidth}{!}{%
\begin{tikzpicture}[
  box/.style    = {draw=FrameGray, line width=1pt},
  wall/.style   = {WallTeal, line width=1.1pt, dash pattern=on 5pt off 4pt,
                   line cap=round},
  string/.style = {AnyonOrange, line width=1.1pt, line cap=round},
  traj/.style   = {AnyonOrange, line width=1.1pt, line cap=round,
                   dash pattern=on 0.1pt off 4.5pt},
  anyon/.style  = {circle, fill=AnyonOrange, inner sep=1.7pt},
  proc/.style   = {-{Stealth[length=3mm]}, FrameGray, line width=1pt},
  ket/.style    = {FrameGray, font=\huge},
  panellabel/.style = {FrameGray, font=\Large, anchor=base west},
  every node/.style = {font=\Large},
]

% ======== part a) : split-and-move sequence ========
\node[panellabel] at (-0.10,5.45) {a)};

% ---- Panel 1 : |psi> ----
\begin{scope}
  \draw[box] (0,0) rectangle (5,5);
  \domainwall
  \node[anyon, label={[AnyonOrange, label distance=-1pt]above right:$b$}] (b1) at (3.90,1.30) {};
  \draw[traj] (b1)
      .. controls (4.30,0.95) and (4.68,0.62) .. (5.00,0.35);
  \node[ket] at (2.50,-0.90) {$|\Psi\rangle$};
\end{scope}

% ---- process arrow: split b ----
\draw[proc] (5.55,2.60) -- (7.25,2.60)
    node[midway, below=3pt, FrameGray] {split $b$};

% ---- Panel 2 ----
\begin{scope}[xshift=7.8cm]
  \draw[box] (0,0) rectangle (5,5);
  \domainwall
  \node[anyon, label={[AnyonOrange, label distance=-1pt]left:$a$}]  (a2) at (3.78,1.44) {};
  \node[anyon, label={[AnyonOrange, label distance=-1pt]above right:$b$}] (b2) at (3.90,1.30) {};
  \draw[string] (a2) .. controls (3.83,1.28) and (3.87,1.40) .. (b2);
  \draw[traj] (b2)
      .. controls (4.30,0.95) and (4.68,0.62) .. (5.00,0.35);
\end{scope}

% ---- process arrow: move a ----
\draw[proc] (13.35,2.60) -- (15.05,2.60)
    node[midway, below=3pt, FrameGray] {move $a$};

% ---- Panel 3 : |phi> ----
\begin{scope}[xshift=15.6cm]
  \draw[box] (0,0) rectangle (5,5);
  \domainwall
  \node[anyon, label={[AnyonOrange, label distance=-1pt]above left:$a$}] (a3) at (1.25,3.60) {};
  \node[anyon, label={[AnyonOrange, label distance=-1pt]above right:$b$}] (b3) at (3.90,1.30) {};
  \draw[string] (b3)
      .. controls (3.25,2.10) and (2.28,3.05) .. (a3);
  \draw[traj] (b3)
      .. controls (4.30,0.95) and (4.68,0.62) .. (5.00,0.35);
  \node[ket] at (2.50,-0.90) {$|\Phi\rangle$};
\end{scope}

% ======== separator ========
\draw[black!12, line width=0.8pt] (21.60,-0.90) -- (21.60,5.45);

% ======== part b) : |phi> with regions u and P_a ========
\node[panellabel] at (22.50,5.45) {b)};

\begin{scope}[xshift=22.6cm]
  \draw[box] (0,0) rectangle (5,5);
  \domainwall

  \coordinate (aP) at (1.25,3.60);

  % annulus P_a (drawn first, so everything sits on top)
  \draw[AnnulusCyan, line width=0.9pt, even odd rule,
        fill=AnnulusCyan, fill opacity=0.14, draw opacity=1]
      (aP) circle (0.95) (aP) circle (0.62);
  \node[AnnulusCyan] at ($(aP)+(45:1.20)+(0.18,0)$) {$P_a$};

  % solid blueish circle u around a
  \draw[DiskBlue, line width=0.9pt,
        fill=DiskBlue, fill opacity=0.18, draw opacity=1]
      (aP) circle (0.50);
  \node[DiskBlue] at ($(aP)+(215:0.27)$) {$u$};

  \node[anyon, label={[AnyonOrange, label distance=-1pt]above right:$b$}] (b4) at (3.90,1.30) {};
  \draw[string] (b4)
      .. controls (3.25,2.10) and (2.28,3.05) .. (aP);
  \draw[traj] (b4)
      .. controls (4.30,0.95) and (4.68,0.62) .. (5.00,0.35);
  \node[anyon, label={[AnyonOrange, label distance=-1pt]above:$a$}] at (aP) {};
\end{scope}

\end{tikzpicture}%
}
	\caption{\small a) The state vector $\ket\Phi$ is prepared from the ground state $\ket\Psi$ in a nonabelian superselection sector by splitting the nonabelian anyon $b$ into $b$ and $a$ and then transporting $a$ into region $A$. The red dotted path indicates a path to infinity, where the $\bar b$ partner of $b$ is located. b) The resulting state cannot be transformed into $\ket\Psi$ by local unitaries in $A$: For every such unitary $u$ there is a loop around its support on which the topological charge $a$ can be detected. The projection $P_a$ for this measurement outcome commutes with $u$, but annihilates $\ket\Psi$.}
  \label{fig:anyons}
\end{figure}

We now provide the physical argument for violation of Haag duality. 
We restrict our argument to renormalization-group fixed points, where anyons can be strictly localized in a finite region, but expect the general argument to be robust against deviations from this assumption. We note that even in the case of renormalization-group fixed points represented by Levin--Wen models, Haag duality for the frustration-free ground state (i.e., without anyons) was only proven very recently \cite{2509.23734,2605.10693}.

Starting from a ground state vector $\ket\Psi$ in a superselection sector hosting a single nonabelian anyon of type $b$ in region $B$ (with its partner at infinity), we produce a state vector $\ket\Phi \in \cH$ that has the same reduced state in $B$, but for every unitary $u$ in $M_A$ we have 
\begin{align}\label{eq:u-orthogonal}
		|\bra\Psi u\ket\Phi| = 0.
\end{align}
The procedure to prepare $\ket\Phi$ is very simple (cf.~\cref{fig:anyons}): 
We split the $b$-anyon into $b$ and a non-vacuum anyon $a$. This is possible, because we assume that $b$ is nonabelian. We then (adiabatically) transport the $a$ anyon into the interior of region $A$.
We assume that we can do this without creating any residual excitations.
Importantly, the final state vector $\ket\Phi$ is, up to a phase, independent of the path along which we transport the anyon to its final destination, as long as we do not encircle the $b$-anyon (recall that there are no other anyons around).
The resulting state $\ket\Phi$ has two essential properties: 
First, its reduced state in region $B$ (i.e., when evaluated with operators in $M_B$) is the same as that of $\ket\Psi$.
The reason is that, given the $b$-type anyon at its fixed location, there is a unique ground state in the $B$ region. Since the system is gapped and we created and moved away $a$ adiabatically, the system remains in this state in the $B$ region.

Second, we need to show \cref{eq:u-orthogonal}. Any unitary in $M_A$ can be approximated arbitrarily well (in the weak operator topology) by a unitary $u$ with finite support in $A$. As emphasized above, no such local unitary $u$ can remove the $a$-anyon from $A$: There is a topological charge measurement that can perfectly distinguish $\ket\Psi$ from $\ket\Phi$ by measuring the topological charge included in the region where $u$ acts. Let $P_a$ be the projector associated to measuring charge $a$ in the region. 
Then $P_a \ket\Phi = \ket\Phi$, while $P_a \ket\Psi = 0$. Since $P_a$ commutes with $u$, we have
\begin{align}
	\bra \Psi u\ket\Phi= \bra\Psi  u P_a \ket\Phi = \bra \Psi P_a u \ket \Phi = 0.
\end{align}
This shows that uniqueness of purifications fails.

We emphasize the neccessity that $b$ is a nonabelian anyon. Indeed, we will prove that for abelian anyon sectors, Haag duality continues to hold when it holds for the ground state and the region $A$ is simply connected. It is known that Haag duality fails if $A$ contains two, infinite connected components \cite{naaijkensKosakiLongoIndexClassification2013,fiedlerJonesIndexSecret2017} and the system allows for abelian (or nonabelian) anyons. 
A simple argument based on uniqueness of purifications can be found in \cite{MR5026085}.  

\subsection{The general theorem}

The idea behind the general result is quite simple and builds upon the framework of \emph{sector theory} \cite{MR297259, MR334742, MR660538, MR2804555, MR3426207, MR4362722, MR4426734, MR4884427, MR4927814, MR4927817, MR4998297, MR4998903, 2410.04736, 2511.08382, 2511.21521, 2603.01936}. 
Mathematically, one way to define an anyon $a$ on the infinite lattice is by a collection of injective maps $a_A \colon M_A \to B(\cH)$ for each bipartition of the system into $A$ and $B$ \cite{MR4927814}. The maps $a_A$ describe how the algebras of local operators change in presence of an $a$-type anyon in $A$: The algebra of observables in $A$ is given by $a_A(M_A)$ and the algebra of observables in $M_B$ is given by $a_B(M_B)$.
Assuming that $M_A = M_B'$, we may assume that $a_A$ is an endomorphism $a_A: M_A\to M_A$ and $a_B$ is the identity on $M_B$ \cite[Facts 2.27]{MR4927814}. 
The abstract fusion product $a \otimes b$ then corresponds to composing the endomorphisms $a_A \circ b_A$.
We then have that $a$ and $\bar a$ always fuse to the vacuum $1$ (i.e., $a$ is abelian) if and only if $\bar a_A$ is an \emph{inverse} of $a_A$. However, $\bar a_A$ is an inverse of $a_A$ if and only if $a_A(M_A) = M_A$. 
Since $a_B$ is taken to be the identity on $M_B$, we have Haag duality,
\begin{align}
	a_A(M_A) = a_B(M_B)',
\end{align}
if and only if $a$ is an abelian anyon. 

\subsection{Levin--Wen models}

We consider the Levin--Wen model built from the unitary fusion category $\cC$ \cite{PhysRevB.71.045110}. 
This model is a commuting projector model \cite{0907.2204, PhysRevB.103.195155, 2305.14068}, and for a finite region $R$, the ground state space over $R$ is isomorphic to the skein module for $R$ \cite{MR3204497, 2305.14068, 2511.21521}. 
In this model, anyon states can be obtained by composing the ground state with an infinite-depth quantum circuit \cite{2511.21521}. 
Each unitary in this circuit acts along a finite part of an infinite string. The full circuit consists of applying these unitaries along an infinite string, with the starting point of one unitary covering the end point of the previous one.
Each unitary consists of four orthogonal partial isometries and an identity on their overall orthogonal complement.
For an anyon type $X$, these partial isometries either create an $X - \bar X$ pair along the string, annihilate such a pair, or move an $X$ excitation from one end of the string to the other. 
To obtain an anyon localized in the region $B$, we take the strings comprising this infinite-depth quantum circuit to be located in $B$. 
We let $\ket{\Psi}$ denote the anyon state obtained in this way. 

Now, suppose $X$ is a nonabelian anyon type and $a \neq 1$ is one of the possible fusion outcomes for $X$ and $\bar X$. 
We can obtain the state in which $a$ is split off from $X$ via the following procedure. 
First, along the first string comprising the infinite-depth quantum circuit, we locally witness $a$ splitting off from $X \otimes \bar X$ (cf.~\cref{fig:LocalAnyonSplitting-LW}). 
We may choose our splitting operator so that the $a$ excitation is located in $A$. 
We now apply our infinite-depth quantum circuit that moves $X$ to infinity. 
We let $\ket{\Phi}$ denote the state obtained in this way. 

\begin{figure}[!ht]
\centering
\begin{tikzpicture}
\draw[thick, dashed] (-5, 1.25) -- (-3.25, -1) -- (-1.5, 1.25);
\filldraw[fill=white] (-3.75, 0) rectangle (-2.75, .5);
\draw[mid>, very thick, blue, rounded corners] (-2.75, .25) -- (-1, -.25) -- (.75, -.25) -- (.75, .25);
\node[blue] at (-2, -.1) {$\scriptstyle a$};
\filldraw[fill=white] (-1, 0) rectangle (0, .5);
\draw[mid>, very thick, red] (0, .25) -- (1, .25);
\node[red] at (.25, .5) {$\scriptstyle X$};
\filldraw[fill=white] (1, 0) rectangle (2, .5);
\filldraw[cyan] (.75, .25) circle (0.05cm);
\node at (-3.25, 1.25) {$\scriptstyle A$};
\node at (-.5, 1.25) {$\scriptstyle B$};
\end{tikzpicture}
	\caption{\small An $a$ anyon being locally split from an $X -\bar X$ anyon pair in the Levin--Wen model. Afterwards, the $X$ anyon is moved to infinity, with an $\bar X$ left behind in $B$.}
\label{fig:LocalAnyonSplitting-LW}
\end{figure}
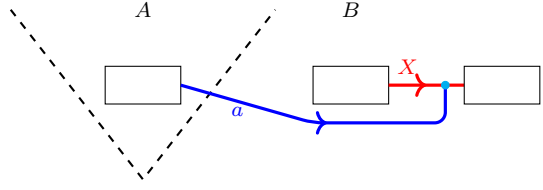

Inside $B$, the two states $\ket{\Psi}$ and $\ket{\Phi}$ have an $\bar X$ excitation where the anyon is located, but otherwise the two states satisfy the ground state constraints within $B$. 
These states therefore agree when evaluated on observables in $B$. 
Indeed, the nonabelian anyon state $\ket{\Psi}$ satisfies a local topological quantum order (LTQO) condition \cite{MR3077916}, as implicitly shown in \cite{2511.21521}. 
On the other hand, the argument in section \ref{sec:PhysicalArgument} shows that for every unitary $u \in M_A$, $\bra{\Psi} u \ket{\Phi} = 0$. 
Hence the argument in section \ref{sec:PhysicalArgument} can be rigorously instantiated in the case of the Levin--Wen model.\footnote{We expect that this argument can also be made rigorous for the nonabelian Quantum Double model, using the anyon states written down in \cite{MR4927817, HamdanThesis}.}

\subsection{Failure of quantum steering} With the absence of uniqueness of purifications, one may ask which parts of pure state entanglement theory remain. In \cite{luijkQuantumSteeringEquivalent2026} it was shown that \emph{quantum steering} \cite{schrodingerProbabilityRelationsSeparated1936,wisemanSteeringEntanglementNonlocality2007} can provide an answer. Given a pure state $\ket\Psi \in \cH$ and local algebras $M_A$,$M_B$, we say that the state of $A$ is \emph{steerable} by $B$ if every ensemble decomposition
\begin{align}
		\omega_A = \sum_x p_x \omega_{A,x}
\end{align}
into a probability distribution $p_x$ and states $\omega_{A,x}$ of the state $\omega_A(a) = \bra\Psi a\ket\Psi$ on $M_A$ can be obtained by a measurement by Bob: There exists a positive operator valued measure $\{b_x\} \subset M_B$ with $0\leq b_x \leq 1$ and $\sum_x b_x = 1$ such that 
\begin{align}
		p_x \omega_{A,x}(a) = \bra\Psi a\, b_x \ket\Psi,\qquad a\in M_A.
\end{align}
It was shown in \cite{luijkQuantumSteeringEquivalent2026} that $A$ is steerable by $B$ if and only if there exists a \emph{conditional expectation} $E\colon M_B' \to M_A$ such that
\begin{align}\label{eq:CE-invariance}
	\bra \Psi E(b) \ket\Psi = \bra\Psi b\ket \Psi,\qquad b\in M_B'. 
\end{align}
A conditional expectation is a unital, completely positive map such that $E^2 = E$. 
We show that, while a conditional expectation $E$ exists, for the case of Levin--Wen models, the anyon state $\ket\Psi$ fulfills \eqref{eq:CE-invariance} if and only if the anyon is abelian. 
Thus, not only uniqueness of purifications, but also steering fails for nonabelian anyon states.  

\subsection{Failure of approximate Haag duality}
As emphasized in the introduction, there is a more general notion of approximate Haag duality that is often used when studying sector theory \cite{MR4362722, 2410.04736, 2511.08382}.\footnote{It is an open question whether approximate Haag duality is in fact more general than strict Haag duality, as there are no known models where approximate Haag duality is known to hold but Haag duality does not hold.}
The advantage of approximate Haag duality is that it has been shown to be stable under perturbations \cite{MR4362722}:
Whether it holds is a property of a gapped phase of matter and not of a single representant of that phase.
Assuming that the vacuum fulfills Haag duality, we show that nonabelian anyon states also violate this more general condition (recall that abelian anyons preserve Haag duality). 
Since a nonabelian anyon state for a Levin--Wen model is a gapped ground state of a commuting projector Hamiltonian \cite{2511.21521} (see also Section~\ref{sec:approximateHD}), our work disproves the conjecture that every gapped ground state satisfies approximate Haag duality.
Since approximate Haag duality is a phase property and is implied by Haag duality, this shows that throughout the whole gapped phase corresponding to this Hamiltonian, Haag duality fails.
This latter gapped phase is of course different from the gapped phase described by the initial Levin-Wen model. Nevertheless, we argue (see remark~\ref{rem:anyons}) that the anyon content of this phase is identical to the initial phase. 

When trying to classify topological orders, we are usually interested in the gapped phase associated to the ``natural vacuum" and not in the gapped phase corresponding to superselection sectors (anyons) on the natural vacuum.
It is a priori not clear how to decide whether a given gapped phase, say represented by a commuting projector Hamiltonian, belongs to the one case or the other.
Our results show that approximate Haag duality (and potentially Haag duality itself) provide a reasonable criterion to select for the right phases. 
This claim is strengthened by the observation that the vacuum phase and the phase associated to a superselection sector sector have the same anyon theories.
One might object that approximate Haag duality cannot distinguish between the vacuum phase and the phase associated to an abelian superselection sector.
However, we expect that the gapped phases corresponding to abelian superselection sectors always coincide with the associated vacuum superselection sectors (in concrete models, \emph{abelian} anyons can be prepared by finite-depth quantum circuits).

\subsection{Conclusion and outlook}
The large-scale structure of entanglement in quantum many-body systems can behave widely different from systems with finitely many degrees of freedom \cite{vanluijkLargeScaleStructureEntanglement2025}. 
Here, we uncovered that a single nonabelian anyon is sufficient to violate the fundamental property of uniqueness of purifications. 
Notably, this occurs even if the subsystems of interest are topologically trivial, in contrast to the case of abelian anyons.

Recently, an alternative approach to topological order has been formulated in terms of the \emph{entanglement bootstrap} program \cite{MR4109024, PhysRevB.103.115150, 2404.05867, 2405.17379, MR4914881}, in which the anyon theory is deduced directly from the entanglement structure of the ground state. It is interesting to observe that a state with a single nonabelian anyon also fails to fulfill the entanglement bootstrap axioms (more precisely, axiom \textbf{A1} \cite{MR4109024}) due to the resulting correction to the topological entanglement entropy for regions containing the anyon \cite{MR2214523, PhysRevLett.96.110405}, while a state with an abelian anyon should still fulfill them. This situation is hence completely analogous to that of Haag duality.
We hence expect that the entanglement bootstrap axioms also imply Haag duality.

\subsection*{Acknowledgments}
We would like to thank Sven Bachmann, Alex Bols, Milo Moses, Pieter Naaijkens, Alexander Stottmeister and Lauritz van Luijk for helpful conversations. 
The conversations with Milo Moses took place at the workshop ``Quantum Field Theory and Topological Phases via Homotopy Theory and Operator Algebras" hosted by the Center of Mathematical Sciences and Applications at Harvard University. 
The conversations with Sven Bachmann and Alex Bols took place at the Les Houches summer school titled ``Quantum Theory on All Scales."

Funded by the European Union. Views and opinions expressed are however those of the authors only and do not necessarily reflect those of the European Union or the European Research Council Executive Agency. Neither the European Union nor the granting authority can be held responsible for them.
This work is supported by an ERC grant (ERC StG LargEnt, 101219447, \href{https://doi.org/10.3030/101219447}{DOI:10.3030/101219447}).

Funded by zukunft.niedersachsen, the joint science funding program of the
Lower Saxony Ministry of Science and Culture and the Volkswagen Foundation.

\subsection*{Use of large lanuage models (LLMs)}
H.W. used a LLM (Claude Opus) for gaining a better understanding of the general theory of nonabelian anyons and sector theory as well as preparing Fig.~\ref{fig:anyons}. 
All technical results have been worked out by the authors, and the entire text of the paper was written by the authors.

\tableofcontents

\section{Basics of operator algebraic approach}
\label{sec:OABasics}

Here, we summarize the operator algebraic approach to quantum spin systems \cite{MR887100, MR1441540, MR3617688}. 
Let $\Gamma$ be an infinite regular lattice on the plane. 
We associate a finite-dimensional Hilbert space $\cH_s \coloneqq \bbC^d$ (with $d \geq 2$) to each site $s \in \Gamma$, where by ``site" we refer to vertices, edges, or faces as appropriate. 
These finite-dimensional Hilbert spaces represent quantum particles at each of the sites. 
To a finite subset $S \subseteq \Gamma$, we associate the Hilbert space $\cH_S \coloneqq \bigotimes_{s \in S} \bbC^d$, which represents the interacting particles. 
We now adopt the Heisenberg approach, where we look at the operators acting on this Hilbert space. 
Specifically, we consider the algebra $\fA(S) \coloneqq B(\cH_S) = \bigotimes_{s \in S} M_d(\bbC)$. 
Note that if $S_1 \subseteq S_2$, then $\fA(S_1) \subseteq \fA(S_2)$, where we include by tensoring with the identity on $S_2 \setminus S_1$. 
We can then take the inductive limit, obtaining the \emph{quasilocal $\rmC^*$-algebra}
\begin{equation}
\fA
\coloneqq 
\overline{\bigotimes_{s \in \Gamma} M_d(\bbC)}^{\| \cdot \|}
=
\overline{\bigcup_{\substack{S \subseteq \Gamma \\ \text{finite}}} \fA(S)}^{\| \cdot \|}.
\end{equation}
For an arbitrary region $\Lambda \subseteq \Gamma$, we can define $\fA(\Lambda) \coloneqq \overline{\bigotimes_{s \in \Lambda} M_d(\bbC)}^{\| \cdot \|}$ similarly. 

A state $\omega \colon \fA \to \bbC$ is a positive unital linear functional. 
To a state $\omega$, we can associate a GNS representation $(\pi, \cH, \Omega)$, where $\pi \colon \fA \to B(\cH)$ is a representation and $\Omega$ is a unit cyclic vector for $\pi(\fA)$ such that $\omega(x) = \bra{\Omega} \pi(x) \ket{\Omega}$ for all $x \in \fA$ \cite[Thm.~2.5.3]{MR3617688}. 
Furthermore, the GNS representation is unique up to unitary equivalence. 
To obtain the GNS representation, one can define a semi-inner product on $\fA$ by $\langle a | b \rangle \coloneqq \omega(a^*b)$. 
Quotienting by the kernel of this form and then completing yields the Hilbert space $\cH$. 
We let $\Omega$ denote the image of $1$ under this quotient map. 
Now, the map $b \mapsto ab$ for $a \in \fA$ is bounded with respect to the form, so it induces an action of $\fA$ on $\cH$, which we denote $\pi \colon \fA \to B(\cH)$. 
One can verify that $(\pi, \cH, \Omega)$ does indeed satisfy the requirements of the GNS representation. 

When studying superselection sectors, we will consider cone regions. 
A \emph{cone in $\bbR^2$} is a region bounded by two rays with the same starting point. 
Note that a cone does not need to be a convex set. 
A \emph{cone in $\Gamma$} is a region of the form $\Lambda \cap \Gamma$, where $\Lambda$ is a cone in $\bbR^2$. 
We will generally use the word ``cone" to refer to a cone in $\Gamma$, not in $\bbR^2$. 

We will also have a fixed, pure reference state $\omega_0$, which in application is a ground state for some model. 
We let $(\pi_0, \cH, \Omega)$ be a GNS representation for $\omega_0$. 
To each cone $\Lambda \subseteq \Gamma$, we associate the von Neumann algebra $\cR(\Lambda) \coloneq \pi_0(\fA(\Lambda))'' \subseteq B(\cH)$. 
Since $\fA(\Lambda)$ and $\fA(\Lambda^c)$ commute, it follows that $\cR(\Lambda) \subset \cR(\Lambda^c)'$. 
We will require the much stronger condition that the net of von Neumann algebras $\cR(\Lambda)$ satisfies \emph{Haag duality}, that is, for every cone $\Lambda$, we have 
\begin{equation}
\cR(\Lambda^c)'
=
\cR(\Lambda).
\end{equation}

\begin{rem}
\label{rem:ProperlyInfinite}
When studying superselection theory, it is also standard to require that the von Neumann algebras $\cR(\Lambda)$ are properly infinite. 
In our setting, the algebra $\cR(\Lambda)$ is \emph{properly infinite} if for every projection $p \in \cR(\Lambda)$ with $p \neq 1$, there exists an isometry $v \in \cR(\Lambda)$ such that $v v^* = p$ \cite[Def.~V.1.15-16]{MR1943006}.\footnote{This definition is equivalent to the usual definition of properly infinite since the von Neumann algebras $\cR(\Lambda)$ are \emph{factors}, meaning that they have trivial center.} 
However, this assumption is satisfied for any gapped ground state of a Hamiltonian with uniformly bounded finite range interactions \cite[Lem.~5.3]{MR4362722}.
\end{rem}

We now define superselection sectors with respect to a reference state $\omega_0$. 
Let $(\pi_0, \cH, \Omega)$ be a GNS representation for $\omega_0$. 
A \emph{superselection sector} is a representation $\pi \colon \fA \to B(\cH_1)$ satisfying that for every cone $\Lambda$, we have 
\begin{equation}
\pi|_{\fA(\Lambda^c)}
\cong
\pi_0|_{\fA(\Lambda^c)},
\end{equation}
where $\cong$ denotes unitary equivalence.
Explicitly, this means that there exists a unitary $U_{\pi,\Lambda}: \cH_1 \to \cH$ such that
\begin{align}
	U_{\pi,\Lambda} \pi(x) U_{\pi,\Lambda}^* = \pi_0(x),\quad x\in \fA(\Lambda^c).
\end{align}

We will often use the term \emph{anyon sector} to refer to an irreducible superselection sector, as the irreducible sectors correspond to the physical anyons.
Here, recall that the representation $\pi$ is called irreducible if its commutant
\begin{align}
		\pi(\fA)' \coloneqq \{ x\in B(\cH_1) : [x,a]=0\ \forall a\in \fA \}
\end{align}
consists only of scalars.

Note that for any cone $\Lambda$, any superselection sector $\pi$ is unitarily equivalent to a sector $\rho \colon \fA \to B(\cH)$ that is \emph{localized in $\Lambda$}, meaning that $\rho|_{\fA(\Lambda^c)} = \pi_0|_{\fA(\Lambda^c)}$. Indeed, we can simply set $\rho = U_{\pi,\Lambda} \pi(\cdot) U_{\pi,\Lambda}^*$.
Evidently, if $V \in \cR(\Lambda)$ is any unitary, then $\tilde \rho(\cdot) = V \rho(\cdot) V^*$ is also a superselection sector localized in $\Lambda$.
The representation $\rho$ is also \emph{transportable}, meaning that for any cone $\Delta$, there exists $\hat \rho \colon \fA \to B(\cH)$ such that $\hat \rho \cong \rho$ and $\hat \rho$ is localized in $\Delta$. 
In summary, a localized and transportable representation $\rho \colon \fA \to B(\cH)$ is a superselection sector, and any superselection sector is unitarily equivalent to a localized and transportable representation $\rho$.

A morphism $T \colon \pi_1 \to \pi_2$ between superselection sectors $\pi_1 \colon \fA \to B(\cH_1)$ and $\pi_2 \colon \fA \to B(\cH_2)$ is an \emph{intertwiner}, that is, a map $T \colon \cH_1 \to \cH_2$ such that $T\pi_1(-) = \pi_2(-)T$. 

If $\pi \colon \fA \to B(\cH_1)$ is a superselection sector, then for any cone $\Lambda$, we have that $\pi|_{\fA(\Lambda)}$ has a unique, normal extension $\pi_\Lambda \colon \cR(\Lambda)\to B(\cH_1)$, where 
\begin{align}
	\pi_\Lambda(\cdot) = V_\Lambda \cdot V_\Lambda^*
\end{align}
for some unitary $V_\Lambda \colon \cH \to \cH_1$ (see for instance \cite[Facts 6.3]{MR4927814}).\footnote{There can only be one normal extension of $\pi|_{\fA(\Lambda)}$ to $\cR(\Lambda)$ since $\pi_0(\fA(\Lambda))$ is dense in $\cR(\Lambda)$ in the $\sigma$-WOT topology.} 
Observe that $\pi_\Lambda$ is necessarily injective.
Additionally, if $\Lambda \subseteq \Delta$, we have that $\pi_\Delta|_{\cR(\Lambda)} = \pi|_\Lambda$ since $\pi_\Delta|_{\cR(\Lambda)}$ is also a normal extension of $\pi|_{\fA(\Lambda)}$ to $\cR(\Lambda)$. 
We can apply this discussion to a superselection sector $\rho$ that is localized in a cone $\Lambda$ to obtain an extension $\rho_\Lambda$ to $\cR(\Lambda)$.
Assuming that the net $\cR(\Lambda)$ satisfies Haag duality, we have the following well-known observations (see \cite[Facts 2.27]{MR4927814}):
\begin{itemize}
\item
If $\rho \colon \fA \to B(\cH)$ is a superselection sector localized in $\Lambda$, then $\rho_\Lambda(\cR(\Lambda)) \subseteq \cR(\Lambda)$,  
\item
If $\rho \colon \fA \to B(\cH)$ is a superselection sector localized in $\Lambda$ and $\Lambda \subseteq \Delta$, then $\rho$ is localized in $\Delta$, 
\item
If $\rho_1 \colon \fA \to B(\cH)$ and $\rho_2 \colon \fA \to B(\cH)$ are superselection sectors localized in $\Lambda$ and $T \colon \rho_1 \to\rho_2$ is an intertwiner, then $T \in \cR(\Lambda)$. 
\end{itemize}

We now define the tensor product of superselection sectors. The most commonly used approach is detailed in \cite{MR4362722} and uses what is known as the auxiliary algebra.\footnote{The paper \cite{MR4927814} provides a different approach for defining the tensor product of superselection sectors that avoids using the auxiliary algebra. This work generalizes the tensor product for conformal nets \cite{MR1231644}. We use the auxiliary algebra here since it is more common for quantum spin systems and it is easier to write down the composition of superselection sectors in this setting.} 
This approach is adapted from the work of \cite{MR660538} using space-like cones in the continuous setting. 

For two cones $\Lambda$ and $\Delta$, we define 
\begin{align}
	\Lambda \Cap \Delta = \{\text{cones } \widehat{\Lambda} : \widehat{\Lambda} \subseteq \Lambda \cap \Delta\}
\end{align}
as the set of cones included in both $\Lambda$ and $\Delta$.
Fix a reference ``excluded" cone $\Lambda_a$. 
The \emph{auxiliary algebra} $\fA^{\Lambda_a}$ is defined to be the C* algebra generated by all the von Neumann algebras $\cR(\Delta)$ for cones $\Delta$ whose intersection with $\Lambda_a$ doesn't contain a cone:
\begin{align}
	\fA^{\Lambda_a} \coloneqq \overline{\bigcup_{\Delta : \Delta \Cap \Lambda_a = \emptyset} \cR(\Delta)}^{\| \cdot \|}.
\end{align}
It is not hard to see that $\pi_0(\fA) \subset \fA^{\Lambda_a}$, independent of the choice of $\Lambda_a$.
Now, let $\rho \colon \fA \to B(\cH)$ be a superselection sector localized in a cone $\Lambda$ such that $\Lambda \Cap \Lambda_a = \emptyset$. 
Then $\rho$ has a unique, norm-continuous extension $\rho^{\Lambda_a}\colon \fA^{\Lambda_a}\to B(\cH)$ such that for every cone $\Delta$ such that $\Delta \Cap \Lambda_a = \emptyset$, we have $\rho^{\Lambda_a}|_{\cR(\Delta)} = \rho_\Delta$ \cite[Lem.~4.1]{MR660538}. 
Furthermore, since $\rho_\Delta(\cR(\Delta)) \subseteq \cR(\Delta)$ whenever $\Lambda \subseteq \Delta$, we have that $\rho^{\Lambda_a}$ is an endomorphism of $\fA^{\Lambda_a}$. 
For $\rho_1, \rho_2 \colon \fA \to B(\cH)$ superselection sectors localized in $\Lambda$, we define $\rho_1 \otimes \rho_2 \coloneqq \rho_1^{\Lambda_a} \circ \rho_2$, which is a superselection sector localized in $\Lambda$ that does not depend on the choice of $\Lambda_a$ \cite[Thm.~4.2]{MR660538}. 

We can also take the tensor product of morphisms. 
Suppose $T \colon \rho_1 \to \rho_2$ and $S \colon \sigma_1 \to \sigma_2$, where $\rho_1, \rho_2, \sigma_1, \sigma_2 \colon \fA \to B(\cH)$ are all superselection sectors localized in some cone $\Lambda$. 
We define
\begin{equation}
T \otimes S
\coloneqq
T(\rho_{1})_{\Lambda}(S)
=
(\rho_{2})_{\Lambda}(S)T.
\end{equation}
Observe that the above expression is well-defined since $T, S \in \cR(\Lambda)$. 
We note the following important special cases of the above tensor product formula: 
\begin{equation}
\id_\rho \otimes S
=
\rho_\Lambda (S), 
\qquad\qquad
T \otimes \id_\sigma
=
T.
\end{equation}

Assuming the von Neumann algebras $\cR(\Lambda)$ satisfy Haag duality and are properly infinite (see Remark \ref{rem:ProperlyInfinite}), the superselection sectors form a braided tensor category \cite{MR4362722}, where the tensor product structure is exactly the one we described. 
In particular, the category admits direct sums and subobjects \cite[Lem.~5.7-5.8]{MR4362722}.

We can draw further parallels to the discussion in the main text if we restrict to the subcategory of dualizable sectors.
A superselection sector $\rho \colon \fA \to B(\cH)$ localized in a cone $\Lambda$ is \emph{dualizable} if there exists another superselection sector $\rho^\vee \colon \fA \to B(\cH)$ localized in $\Lambda$ and morphisms $R \colon \mathds{1} \to \rho^\vee \otimes \rho$, $\bar R \colon \mathds{1} \to \rho \otimes \rho^\vee$ that satisfy the following \emph{zig-zag equations}: 
\begin{equation}
\label{eq:zig-zag}
(R^* \otimes \id_{\rho^\vee})(\id_{\rho^\vee} \otimes \bar R)
=
R^* (\rho^\vee)_\Lambda(\bar R) = \mathds{1},
\qquad\qquad
(\id_\rho \otimes R^*)(\bar R \otimes \id_\rho)
=
\rho_\Lambda (R^*)\bar R 
= 
\mathds{1}.
\end{equation}
These equations are called the ``zig-zag equations" because they have a graphical interpretation. 
We can graphically depict the morphisms $R, \bar R$ as follows: 
\begin{equation}
R
\coloneqq
\tikzmath{
\draw[mid<] (0, .5) arc(0:-180:.5);
},
\qquad\qquad
\bar R
\coloneqq
\tikzmath{
\draw[mid>] (0, .5) arc(0:-180:.5);
}.
\end{equation}
The zig-zag equations can then be written in the following way: 
\begin{equation}
\tikzmath{
\draw[mid<] (0, -.25) -- (0, .75) arc(180:0:.375) arc(-180:0:.375) -- (1.5, 1.75);
}
=
\tikzmath{
\draw[mid<] (0, -.25) -- (0, 1.75);
}\,,
\qquad\qquad
\tikzmath{
\draw[mid>]
(0, 1.75) -- (0, .75) arc(-180:0:.375) arc(180:0:.375) -- (1.5, -.25);
}
=
\tikzmath{
\draw[mid>] (0, -.25) -- (0, 1.75);
}\,.
\end{equation}

For an anyon sector $\rho$, the sector $\rho^\vee$ corresponds exactly to the ``dual anyon type" discussed in section the main text. 
We call $\rho^\vee$ a \emph{conjugate sector} for $\rho$. 
The dualizable sectors form a braided tensor subcategory of the category of all superselection sectors \cite[Thm.~2.4]{MR1444286}. 
Furthermore, because the category of superselection sectors admits subobjects, every dualizable sector decomposes into a direct sum of irreducible ones \cite[Lem.~3.2]{MR1444286}. 
These facts allow us to recover the discussion in the main text in this framework. 
Bringing two anyons with associated sectors $\pi, \sigma$ ``close together" corresponds to the sector $\pi \otimes \sigma$. 
Assuming $\pi, \sigma$ are dualizable, $\pi \otimes \sigma$ is also dualizable, so we can write $\pi \otimes \sigma$ as a direct sum of (irreducible) anyon sectors. 
This direct sum decomposition corresponds to the different possible fusion outcomes and the associated fusion channels. 
In all known models, every anyon sector is dualizable.

\section{DHR argument}

In the following let $\rho$ be an anyon sector localized in a cone $\Lambda$. 
We then obtain a pure state $\omega_\rho$ on $\fA$ by setting
\begin{align}
		\omega_\rho(x) = \bra\Omega \rho(x)\ket\Omega.
\end{align}
The state $\omega_\rho$ describes a situation where an anyon of type $\rho$ is localized in $\Lambda$, but the state coincides with the vacuum $\omega_0$ outside of $\Lambda$.
Denote by $(\pi_\rho, \Omega_\rho,\cH_\rho)$ a GNS triple of $\omega_\rho$ and set
\begin{align}
	\cR_\rho(\Delta) = \pi_\rho(\fA(\Lambda))''
\end{align}
for any cone $\Delta$. 

\begin{thm}\label{thm:dhr}
Suppose that the net of reference von Neumann algebras $\cR(\Lambda)$ fulfills Haag duality and let $\rho: \fA \to B(\cH)$ be an anyon sector localized in a cone $\Lambda$. The following are equivalent:
	\begin{enumerate}
		\item\label{item:dhr1} $\rho$ is abelian, i.e., there exists a superselection sector $\rho^\vee$ localized in $\Lambda$ such that $\rho\otimes \rho^\vee = \rho^\vee \otimes \rho = \mathds{1}$.
		\item\label{item:dhr2} The net of von Neumann algebras $\cR_\rho(\Delta)$ fulfills Haag duality.
	\end{enumerate}
\end{thm}

We begin by clarifying the relationship between $\cR_\rho(\Delta)$ and $\cR(\Delta)$. 
Recall from the previous section that if $\rho \colon \fA \to B(\cH)$ is a superselection sector, then for any cone $\Delta$, we have that $\rho|_{\fA(\Delta)}$ has a unique, normal and injective extension $\rho_\Delta$ to $\cR(\Delta)$, where $\rho_\Delta = \Ad(V_\Delta)$ for some unitary $V_\Delta \in B(\cH)$. 
Since $\rho$ is irreducible, $(\rho, \Omega,\cH)$ is also a GNS triple for $\omega_\rho$. By the uniqueness of the GNS representation, there exists a unitary $U \colon \cH \to \cH_\rho$ such that
\begin{align}
		\cR_\rho(\Delta) = U\rho(\fA(\Delta))''U^* = U\rho_\Delta(\cR(\Delta))U^*
\end{align}
for any cone $\Delta$. 
Thus Haag duality for the net $\cR_\rho(\Delta)$ is equivalent to Haag duality for the net $\rho_\Delta(\cR(\Delta))$.

To show \cref{thm:dhr}, we adapt the argument in \cite[Lem.~2.2]{MR297259} to the quantum spin system setting, using the setup originally formulated in \cite{MR660538}. 

\begin{proof}[Proof of \cref{item:dhr2} $\Rightarrow$ \cref{item:dhr1}]
For any cone $\Delta$ with $\Lambda \subseteq \Delta$, we have that $\rho$ is localized in $\Delta$ and thus
\begin{equation}
\rho_\Delta(\cR(\Delta))
=
\rho_{\Delta^c}(\cR(\Delta^c))'
=
\cR(\Delta^c)'
=
\cR(\Delta).
\end{equation}
Therefore, $\rho_\Delta$ defines an automorphism of $\cR(\Delta)$. 
Note that $\rho_\Delta^{-1} = \Ad(V_\Delta^*)$, so $\rho_\Delta^{-1}$ is also normal. 

Let $\Lambda_a$ be a cone disjoint from $\Lambda$.
By the same argument used to define $\rho^{\Lambda_a}$ from $\rho$, the collection of automorphisms $\rho_\Delta^{-1} \colon \cR(\Delta) \to \cR(\Delta)$ extends uniquely to a norm-continuous endomorphism $(\rho^\vee)^{\Lambda_a}$ of $\fA^{\Lambda_a}$.
For any cone $\Delta$ such that $\Delta \Cap \Lambda_a = \emptyset$, we have that 
\begin{gather}
(\rho^\vee)^{\Lambda_a} \circ \rho^{\Lambda_a}|_{\cR(\Delta)}
=
\rho_\Delta^{-1} \circ \rho_\Delta
=
\id_{\cR(\Delta)},
\\
\rho^{\Lambda_a} \circ (\rho^\vee)^{\Lambda_a}|_{\cR(\Delta)}
=
\rho_\Delta \circ \rho_\Delta^{-1}
=
\id_{\cR(\Delta)}.
\end{gather}
Therefore, since $\fA^{\Lambda_a} \coloneqq \overline{\bigcup_{\Delta : \Delta \Cap \Lambda_a = \emptyset} \cR(\Delta)}^{\| \cdot \|}$, we have that $(\rho^\vee)^{\Lambda_a} = (\rho^{\Lambda_a})^{-1}$.

We now define $\rho^\vee \coloneqq (\rho^\vee)^{\Lambda_a}|_{\fA}$. 
We claim that $\rho^\vee \colon \fA \to B(\cH)$ is a superselection sector. 
If this is the case, $\rho^\vee$ will necessarily be the desired superselection sector, since $\rho \otimes \rho^\vee = \rho^{\Lambda_a} \circ \rho^\vee$ and $\rho^\vee \otimes \rho = (\rho^\vee)^{\Lambda_a} \circ \rho$. 

The proof that $\rho^\vee$ is a superselection sector proceeds similarly to \cite[Thm.~4.2]{MR660538}.
We first show that $\rho^\vee$ is independent of the choice of $\Lambda_a$. 
First consider the case of two cones $\Lambda_a$ and $\widehat{\Lambda}_a$ such that $\Lambda_a \subseteq \widehat{\Lambda}_a$.
We then have that $\fA^{\widehat{\Lambda}_a} \subseteq \fA^{\Lambda_a}$, so $(\rho^\vee)^{\widehat{\Lambda}_a} = (\rho^\vee)^{\Lambda_a}|_{\fA^{\widehat{\Lambda}_a}}$. 
	Since $\fA \subset \fA^{\Lambda_a} \cap \fA^{\widehat{\Lambda_a}}$ it thus follows that $(\rho^\vee)^{\widehat{\Lambda}_a}|_{\fA} = (\rho^\vee)^{\Lambda_a}|_{\fA}$.

	We now extend the argument to two arbitrary cones $\Lambda_a$ and $\widehat{\Lambda}_a$ satisfying $\Lambda_a \cap \Lambda = \emptyset = \widehat{\Lambda}_a \cap \Lambda$. Observe that in this case there exists a sequence of cones $((\Lambda_a)_1, (\Delta_a)_1, \dots, (\Lambda_a)_{n - 1}, (\Delta_a)_{n - 1}, (\Lambda_a)_n)$ such that $(\Lambda_a)_1 = \Lambda_a$, $(\Lambda_a)_n = \widehat{\Lambda}_a$, and for all $i$, we have $(\Lambda_a)_i, (\Lambda_a)_{i + 1} \subseteq (\Delta_a)_i$. 
We call such a sequence a \emph{zig-zag} following \cite{MR4927814}, but such a sequence is termed an ``interpolating sequence" in \cite{MR660538}.
We further require that $(\Delta_a)_i \cap \Lambda = \emptyset$ for all $i$.
By the preceeding argument, we have
	\begin{align}
		(\rho^\vee)^{(\Lambda_a)_i}|_{\fA} = (\rho^\vee)^{(\Delta_a)_i}|_{\fA} = (\rho^\vee)^{(\Lambda_a)_{i+1}}|_{\fA}.
	\end{align}
	It thus follows that $(\rho^\vee)^{\Lambda_a}|_{\fA} = (\rho^\vee)^{\widehat{\Lambda}_a}|_{\fA}$.
 	Hence $\rho^\vee$ is independent of the choice of auxiliary cone $\Lambda_a$.

We now show that $\rho^\vee$ is actually a superselection sector. 
Let $\widehat{\Lambda}$ be another cone. 
Since $\rho$ is a superselection sector, there exists a superselection sector $\hat \rho \colon \fA \to B(\cH)$ localized in $\widehat{\Lambda}$ such that $\rho \cong \hat \rho$, where $\cong$ denotes unitary equivalence. 
Now, since $\rho \cong \hat \rho$, the net of von Neumann algebras $\hat \rho_\Delta(\cR(\Delta))$ also satisfy Haag duality, so we can repeat the prior analysis for $\hat \rho$ and obtain a map $\hat \rho^\vee \colon \fA \to B(\cH)$. 
We claim that $\rho^\vee \cong \hat \rho^\vee$. 
Note that there is a zig-zag from $\Lambda$ to $\widehat{\Lambda}$, so as before, it suffices to consider the case where $\Lambda \subseteq \widehat{\Lambda}$. 
Let $V \colon \rho \to \hat \rho$ be a unitary intertwiner. 
Then since the original net $\cR(\Lambda)$ satisfies Haag duality and $\rho, \hat \rho$ are both localized in $\widehat{\Lambda}$, we have that $V \in \cR(\widehat\Lambda)$. 
Now, let $\Lambda_a$ be a cone such that $\Lambda_a \cap \widehat{\Lambda} = \emptyset$. 
Then $V \in \fA^{\Lambda_a}$. 
Since $V \colon \rho \to \hat \rho$, we have that $V \colon \rho^{\Lambda_a} \to \hat \rho^{\Lambda_a}$, and thus $\hat \rho^{\Lambda_a} = \Ad(V) \circ \rho^{\Lambda^a}$. 
Therefore, we have that 
\begin{equation}
(\hat \rho^\vee)^{\Lambda_a}
=
(\hat \rho^{\Lambda_a})^{-1}
=
(\rho^{\Lambda_a})^{-1} \circ \Ad(V^*)
=
(\rho^\vee)^{\Lambda_a} \circ \Ad(V^*)
=
\Ad [(\rho^\vee)^{\Lambda_a}(V^*)] \circ (\rho^\vee)^{\Lambda_a},
\end{equation}
where in the last step we used that $V^* \in \fA^{\Lambda_a}$.
Since $(\rho^\vee)^{\Lambda_a}$ is an endomorphism of $\fA^{\Lambda_a}$ and $V^*$ is unitary, it follows that $(\rho^\vee)^{\Lambda_a}(V^*)$ is also unitary. 
Restricting to $\fA$, we obtain that $\hat \rho^\vee \cong \rho^\vee$, as desired. 
\end{proof}

\begin{rem}
Note that there is a formulation of superselection theory using nets of von Neumann algebras that is more general than the approach used here and does not require the auxiliary algebra \cite{MR4927814}. 
In the quantum spin system setting, these approaches are equivalent \cite[Cor.~6.7]{MR4927814}. 
We expect that this argument can be adapted to the more general setting of \cite{MR4927814}; however, since this paper focuses on quantum spin systems, this is beyond the scope of our work. 
\end{rem}

The converse of our main result similarly follows from the proof of \cite[Lem.~2.2]{MR297259} adapted to our setting, but we spell it out here for completeness.
\begin{proof}[Proof of \cref{item:dhr1} $\Rightarrow$ \cref{item:dhr2}]
	Suppose $\rho \colon \fA \to B(\cH)$ is an \emph{abelian} anyon localized in the cone $\Lambda$. 
We show that the net of von Neumann algebras $\rho_\Delta(\cR(\Delta))$ satisfy Haag duality in this case. 
Since $\rho$ is an abelian anyon, there exists an anyon sector $\rho^\vee \colon \fA \to B(\cH)$ localized in $\Lambda$ such that $\rho \otimes \rho^\vee = \mathds{1}$ and $\rho^\vee \otimes \rho = \mathds{1}$. 
In particular, for any cone $\Lambda_a$ satisfying $\Lambda_a \Cap \Lambda = \emptyset$, we have that $\rho^{\Lambda_a}$ defines an automorphism of the auxiliary algebra $\fA^{\Lambda_a}$, whose inverse is $(\rho^\vee)^{\Lambda_a}$. 
Since $\rho$ and $\rho^\vee$ are both localized in $\Lambda$, we have that $\rho_\Lambda$ and $(\rho^\vee)_{\Lambda}$ are both endomorphisms of $\cR(\Lambda)$. 
It follows that $(\rho^\vee)_\Lambda = (\rho_\Lambda)^{-1}$ since $(\rho^\vee)^{\Lambda_a} = (\rho^{\Lambda_a})^{-1}$ and $\cR(\Lambda) \subseteq \fA^{\Lambda_a}$.
In particular, $\rho_\Lambda \colon \cR(\Lambda) \to \cR(\Lambda)$ is an automorphism, so $\rho_\Lambda(\cR(\Lambda)) = \cR(\Lambda)$. 
Since $\rho$ is localized in $\Lambda$, we have that $\rho_{\Lambda^c} = \id_{\cR(\Lambda^c)}$, so we have that 
\begin{equation}
\rho_{\Lambda^c}(\cR(\Lambda^c))'
=
\cR(\Lambda^c)'
=
\cR(\Lambda)
=
\rho_{\Lambda}(\cR(\Lambda)).
\end{equation}

To see that Haag duality holds for all cones, note that for any other cone $\Delta$, we can transport $\rho$ to obtain a unitarily equivalent sector $\hat \rho \colon \fA \to B(\cH)$ localized in $\Delta$. 
We then have $\hat \rho_\Delta(\cR(\Delta^c))' = \hat \rho_\Delta(\cR(\Delta))$. 
But letting $V \colon \hat \rho \to \rho$ be a unitary witnessing the equivalence, we have that 
\begin{equation}
\rho_{\Delta^c}(\cR(\Delta^c))'
=
(V\hat \rho_{\Delta^c}(\cR(\Delta^c))V^*)'
=
V\hat \rho_{\Delta^c}(\cR(\Delta^c))'V^*
=
V\hat \rho_{\Delta}(\cR(\Delta))V^*
=
\rho_\Delta(\cR(\Delta)).
\end{equation}
Hence the net of algebras $\rho_\Delta(\cR(\Delta))$ satisfies Haag duality. 
\end{proof}
\section{Rigorous version of physical argument for string-net models}
\label{sec:LWRigorousHeuristic}

In this section, we provide a rigorous interpretation of our physical argument for Levin--Wen string-net models \cite{PhysRevB.71.045110}. 
In sections \ref{subsec:LWBasics}--\ref{subsec:LWAnyonSectors}, we summarize the setup of \cite{2511.21521}, which classifies anyon sectors for the Levin--Wen model in the infinite volume. 
In section \ref{subsec:HeuristicMadeRigorous}, we use this setup to show that the heuristic holds in this mathematically rigorous framework. 

\subsection{The Levin--Wen model}
\label{subsec:LWBasics}
We use the setup described in \cite{2511.21521} for the Levin--Wen string-net models \cite{PhysRevB.71.045110}.
Let $\cC$ be a unitary fusion category, and let $\Irr(\cC)$ denote a set of representative simple objects for $\cC$. 
Define $\chi \coloneqq \bigoplus_{c \in \Irr(\cC)} c$. 
For $c \in \cC$ we let $d_c$ denote the quantum dimension of $c$. 
We consider a regular square lattice $\Gamma$, whose vertices, edges, and faces we denote by $V(\Gamma)$, $E(\Gamma)$, and $F(\Gamma)$ respectively.
To each vertex $v \in V(\Gamma)$, we associate the Hilbert space $\cH_v \coloneqq \End_\cC(\chi \otimes \chi)= \bigoplus_{a, b, c, d \in \Irr(\cC)} \cC(a \otimes b \to c \otimes d)$. 
We graphically depict a vector of the form $f \colon a \otimes b \to c \otimes d$ by 
\begin{equation}
\tikzmath{
\draw (0, 0) -- (1.5, 0);
\draw (.75, -.75) -- (.75, .75);
\draw[cyan, dashed] (.25, .5) -- (1.25, -.5);
\roundNbox{fill=white}{(.75, 0)}{.3}{0}{0}{$f$};
\node at (-.1, 0) {$\scriptstyle a$};
\node at (.75, -.9) {$\scriptstyle b$};
\node at (.75, .85) {$\scriptstyle c$};
\node at (1.6, 0) {$\scriptstyle d$};
}
\end{equation}
where we read the diagram from left to right and bottom to top as indicated by the dashed cyan line. 
To make $\cH_v$ into a Hilbert space, we endow it with the \emph{skein module inner product}, that is, for $f\colon a_1 \otimes b_1 \to c_1 \otimes d_1$ and $g \colon a_2 \otimes b_2 \to c_2 \otimes d_2$, we define
\begin{equation}
\label{eq:HvSkeinModuleInnerProduct}
\langle g, f \rangle
\coloneqq
\delta_{a_1 = a_2} \delta_{b_1 = b_2} \delta_{c_1 = c_2} \delta_{d_1 = d_2}
\frac{1}{\sqrt{d_{a_1}d_{b_1} d_{c_1} d_{d_1}}}
\tikzmath{
\draw (-.1, -.2) -- (-.1, 1.45) arc(180:0:.5) -- (.9, -.2) arc(0:-180:.5);
\draw (.1, -.2) -- (.1, 1.45) arc(180:0:.3) -- (.7, -.2) arc(0:-180:.3);
\roundNbox{fill=white}{(0, 0)}{.3}{0}{0}{$f$};
\roundNbox{fill=white}{(0, 1.25)}{.3}{0}{0}{$g^\dag$};
} \, .
\end{equation}
We note for later use that the \emph{trace inner product} $(g,f)$ is obtained when the normalization by the square-roots of quantum dimensions is omitted, i.e., if $g,f \in \cC(a\to b)$, then $(g,f) = \tr(g^\dagger \circ f)$, where $\tr$ is the spherical trace fulfilling $\tr(\id_a) = d_a$.

For $S \subseteq V(\Gamma)$ finite, we define $\cH_S \coloneqq \bigotimes_{v \in S} \cH_v$ and $\fA(S) \coloneqq B(\cH_S) = \bigotimes_{v \in S} B(\cH_v)$.
Our quasilocal algebra is then 
\begin{equation}
\fA 
\coloneqq 
\overline{\bigotimes_{v \in V(\Gamma)} B(\cH_v)}^{\| \cdot \|}
=
\overline{\bigcup_{\substack{S \subseteq V(\Gamma) \\ \text{finite}}} \fA(S)}^{\| \cdot\|}
\end{equation}
We also define $\fA(\Lambda) \coloneqq \overline{\bigotimes_{v \in \Lambda} B(\cH_v)}^{\| \cdot \|}$ for arbitrary $\Lambda \subseteq V(\Gamma)$.

For the Levin--Wen model, we define two types of Hamiltonian terms. 
For each edge $e \in E(\Gamma)$, we define the projection $A_e$ on simple tensors by 
\begin{equation}
A_e \left( \tikzmath{
\draw (0, 0) -- (1.5, 0);
\draw (.75, -.75) -- (.75, .75);
\draw[dashed, cyan] (1.85, 0) -- (2.35, 0);
\draw (2.75, 0) -- (4.25, 0);
\draw (3.5, -.75) -- (3.5, .75);
\roundNbox{fill=white}{(.75, 0)}{.3}{0}{0}{$f$};
\roundNbox{fill=white}{(3.5, 0)}{.3}{0}{0}{$g$};
\node at (-.15, 0) {$\scriptstyle a_1$};
\node at (.75, -.9) {$\scriptstyle b_1$};
\node at (.75, .85) {$\scriptstyle c_1$};
\node at (1.7, 0) {$\scriptstyle d_1$};
\node at (2.575, 0) {$\scriptstyle a_2$};
\node at (3.5, -.9) {$\scriptstyle b_1$};
\node at (3.5, .85) {$\scriptstyle c_1$};
\node at (4.45, 0) {$\scriptstyle d_2$};
\node[cyan] at (2.1, .2) {$\scriptstyle e$};
} \right)
=
\delta_{d_1 = a_2} 
\tikzmath{
\draw (0, 0) -- (3.25, 0);
\draw (.75, -.75) -- (.75, .75);
\draw (2.5, -.75) -- (2.5, .75);
\roundNbox{fill=white}{(.75, 0)}{.3}{0}{0}{$f$};
\roundNbox{fill=white}{(2.5, 0)}{.3}{0}{0}{$g$};
\node at (-.15, 0) {$\scriptstyle a_1$};
\node at (.75, -.9) {$\scriptstyle b_1$};
\node at (.75, .85) {$\scriptstyle c_1$};
\node at (2.5, -.9) {$\scriptstyle b_1$};
\node at (2.5, .85) {$\scriptstyle c_1$};
\node at (3.45, 0) {$\scriptstyle d_2$};
}
\end{equation}
where we connect the edge in the image of $A_e$ to show that the 1-skeleton has been condensed. 
We now define the projection $B_f$ for $f \in F(\Gamma)$. 
The projection $B_f$ acts on $\prod_{e \in f} A_e \cH_f$. 
Here, we write $e \in f$ to denote that the edge $e$ is adjacent to the face $f$, and we have $\cH_f \coloneqq \bigotimes_{v \in f} \cH_v$, where $v \in f$ means that $v$ is adjacent to an edge adjacent to $f$. 
Specifically, $B_f$ acts on simple tensors as follows: 
\begin{equation}
B_f
\left( \tikzmath{
\draw (-.75, .5) -- (.75, .5);
\draw (-.5, .75) -- (-.5, -.75);
\draw (-.75, -.5) -- (.75, -.5);
\draw (.5, .75) -- (.5, -.75);
\node[gray] at (0, 0) {$\scriptstyle \times$};
}\right)
=
\frac{1}{D_\cC^2} \sum_{c \in \Irr(\cC)} d_c
\tikzmath{
\draw (-.75, .5) -- (.75, .5);
\draw (-.5, .75) -- (-.5, -.75);
\draw (-.75, -.5) -- (.75, -.5);
\draw (.5, .75) -- (.5, -.75);
\node[gray] at (0, 0) {$\scriptstyle \times$};
\draw[rufous] (0, 0) circle (.4);
\node[rufous] at (.275, 0) {$\scriptstyle c$};
}\, ,
\end{equation}
where $D_\cC$ denotes the global dimension of $\cC$.
To interpret the diagram on the right, we use the fusion relation 
\begin{equation}
\label{eq:FusionRelation}
\tikzmath{
\draw (0, 0) -- (0, 1);
\draw (.25, 0) -- (.25, 1);
\node at (0, -.15) {$\scriptstyle a$};
\node at (.25, -.15) {$\scriptstyle b$};
}
=
\sum_{c \in \Irr(\cC)}
\frac{\sqrt{d_c}}{\sqrt{d_a d_b}}
\tikzmath{
\draw (0, 0) -- (.125, .25) -- (.25, 0);
\draw[rufous] (.125, .25) -- (.125, .75);
\draw (0, 1) -- (.125, .75) -- (.25, 1);
\filldraw[cyan] (.125, .25) circle (0.04cm);
\filldraw[cyan] (.125, .75) circle (0.04cm);
\node at (0, -.15) {$\scriptstyle a$};
\node at (.25, -.15) {$\scriptstyle b$};
\node[rufous] at (.225, .5) {$\scriptstyle c$};
}
\end{equation}
where the colored vertices denote summing over an orthonormal basis for the skein module (c.f.~\eqref{eq:HvSkeinModuleInnerProduct}) and its dual. 
It is well known that the operators $B_f$ are projections and that for any two $f_1, f_2 \in F(\Gamma)$, we have $[B_{f_1}, B_{f_2}] = 0$ \cite{0907.2204, PhysRevB.103.195155, 2305.14068}. 
We define the Levin--Wen Hamiltonian $H$ as follows: for a finite region $R \subseteq \Gamma$, we define 
\begin{equation}
H_R
\coloneqq
\sum_{e \subset R} (\mathds{1} - A_e) + \sum_{f \subset R} (\mathds{1} - B_f).
\end{equation}
For a finite region $R \subseteq \Gamma$, we define $p_R \coloneqq \prod_{f \in R} B_f \prod_{e \in R} A_e$, which is the projection onto the ground state space for the region $R$. 
There is a unique frustration-free ground state $\omega_0$ for the Levin--Wen Hamiltonian, which is a pure state \cite{MR4945955, 2511.21521}.
In particular, $\omega_0$ is the unique state satisfying that $\omega_0(p_R) = 1$ for all finite $R \subseteq \Gamma$. 
Furthermore, the state $\omega_0$ is known to satisfy Haag duality for cones \cite{2509.23734, 2605.10693}.
We let $(\pi_0, \cH, \Omega)$ denote the GNS representation of $\omega_0$. 

We now characterize the Hilbert space $p_R \cH_R$ for a finite, connected region $R$. 
We first define the \emph{skein module} $\cS(R)$. 
We use $\partial R$ to denote the edges in $E(\Gamma)$ transverse to the boundary of $R$. 
We define $\cS(R)$ to be the vector space of string diagrams for $\cC$ in $R$ with marked boundary points along $\partial R$, up to isotopy and relations in $\cC$. 
Each marked boundary point is labeled with the object $\chi = \bigoplus_{c \in \Irr(\cC)} c$. 
Equivalently, we can think of elements in $\cS(R)$ as sums of elements where each marked boundary point is labeled with an object $c \in \Irr(\cC)$. 
The space $\cS(R)$ can be endowed with an inner product (see \cite[\textsection 3.6]{2511.21521}). 
Furthermore, there is a unitary isomorphism $\sigma_R \colon p_R \cH_R \to \cS(R)$ \cite[Prop.~4.4]{2511.21521}.\footnote{In the case where $R$ is simply connected, a proof of this result is given in \cite[Thm.~2.9]{2305.14068}. The result was stated without proof in \cite{MR3204497}.}
For our purposes, we will not require the specific formula for this unitary isomorphism, nor the actual inner product on $\cS(R)$. 
However, the formulas we present later will be written to ensure that all actions on these Hilbert spaces are $*$-preserving.

\subsection{Tube algebras and skein modules}
Consider a cylinder $A = S^1 \times I$ where we decorate $S^1 \times \{0\}$ and $S^1 \times \{1\}$ with a boundary configuration $\cB$, i.e., a collection of marked points labeled by $\chi = \bigoplus_{c \in \Irr(\cC)} c$.
We can then define the skein module $\cS(A)$ as above as the vector space of string diagrams on $A$ compatible with the boundary configuration, up to isotopy and relations in $\cC$. 
As before, we can alternatively think of elements of $\cS(A)$ as sums of elements where each marked boundary point is labeled by some $c \in \Irr(\cC)$ instead of by $\chi$. 
We can turn $\cS(A)$ into an algebra where the composition is induced from stacking the cylinders:
\begin{equation}
\tikzmath{
\draw[very thick] (0, -.25) ellipse (1 and .25);
\draw[very thick] (1, -1.75) arc (0:-180:1 and .25);
\draw (1, -1) arc (0:-180:1 and .25);
\draw[dotted, black!50] (-1, -1) arc(180:0:1 and .25);
\draw (1, -1.75) -- (1, -.25);
\draw (-1, -1.75) -- (-1, -.25);
\draw (0, -2) -- (0, -.5);
\roundNbox{fill=white}{(0, -1.25)}{.3}{0}{0}{$g$}
\node at (.2, -1.75) {$\scriptstyle \vec{a}_3$};
\node at (.2, -.7) {$\scriptstyle \vec{a}_4$};
\node at (-.75, -1) {$\scriptstyle c$};
}
\cdot
\tikzmath{
\draw[very thick] (0, -.25) ellipse (1 and .25);
\draw[very thick] (1, -1.75) arc (0:-180:1 and .25);
\draw (1, -1) arc (0:-180:1 and .25);
\draw[dotted, black!50] (-1, -1) arc(180:0:1 and .25);
\draw (1, -1.75) -- (1, -.25);
\draw (-1, -1.75) -- (-1, -.25);
\draw (0, -2) -- (0, -.5);
\roundNbox{fill=white}{(0, -1.25)}{.3}{0}{0}{$f$}
\node at (.2, -1.75) {$\scriptstyle \vec{a}_1$};
\node at (.2, -.7) {$\scriptstyle \vec{a}_2$};
\node at (-.75, -1) {$\scriptstyle d$};
}
=
\delta_{\vec{a}_2 = \vec{a}_3}
\tikzmath{
\draw[very thick] (0, -.25) ellipse (1 and .25);
\draw (1, -1) arc (0:-180:1 and .25);
\draw[dotted, black!50] (-1, -1) arc(180:0:1 and .25);
\draw (1, -1.75) -- (1, -.25);
\draw (-1, -1.75) -- (-1, -.25);
\draw (0, -2) -- (0, -.5);
\draw[very thick] (1, -2.75) arc (0:-180:1 and .25);
\draw (1, -2) arc (0:-180:1 and .25);
\draw[dotted, black!50] (-1, -2) arc(180:0:1 and .25);
\draw (1, -2.75) -- (1, -1.25);
\draw (-1, -2.75) -- (-1, -1.25);
\draw (0, -3) -- (0, -1.5);
\roundNbox{fill=white}{(0, -1.25)}{.3}{0}{0}{$g$}
\node at (.2, -.7) {$\scriptstyle \vec{a}_4$};
\node at (-.75, -1) {$\scriptstyle c$};
\roundNbox{fill=white}{(0, -2.25)}{.3}{0}{0}{$f$}
\node at (.2, -2.75) {$\scriptstyle \vec{a}_1$};
\node at (.2, -1.75) {$\scriptstyle \vec{a}_2$};
\node at (-.75, -2) {$\scriptstyle d$};
},
\end{equation}
In the above equation, each label corresponds to an object in $\Irr(\cC)$. 
Note that we can simplify the two parallel $c$ and $d$ strands by applying \eqref{eq:FusionRelation}.

This algebra is called the \emph{tube algebra} \cite{MR1782145, MR1832764, MR1966525}, and we denote it by $\Tube_\cB$. 
Observe that for any finite connected region $R \subseteq \Gamma$ with decorated boundary components $\cB_1, \dots, \cB_n$, we have an action of $\Tube_{\cB_k}$ on $\cS(R)$ by gluing along the boundary component $\cB_k$. 
This action can be transported to an action of $\Tube_{\cB_k}$ on $p_R\cH_R$ via the isomorphism $\sigma_R$. 
For $a \in \Tube_{\cB_k}$, we write $\mathfrak{t}_{\cB_k}(a) \in B(p_R\cH_R)$ to denote the action of $a$ along the boundary component $\cB_k$. 

We now construct projections in $\Tube_\cB$ that correspond to anyon types. 
Recall that anyon types for the Levin--Wen model correspond to simple objects in $Z(\cC)$, the Drinfeld center of $\cC$ \cite{PhysRevB.71.045110,kitaevModelsGappedBoundaries2012,PhysRevB.103.195155, 2305.14068}.
This physical picture was made mathematically rigorous for Levin--Wen models in \cite{2511.21521}. 
The objects of the Drinfeld center $Z(\cC)$ consist of pairs $(X,\gamma_{X, \placeholder})$ of an object $X\in \cC$ and a \emph{half-braiding} $\gamma$. The half-braiding associates to every $Y\in \cC$ an isomorphism $\gamma_{X, Y} \colon X\otimes Y \to Y\otimes X$ that is natural in $Y$ and fulfills the consistency condition $\gamma_{X, Y\otimes Z} = (\id_Y \otimes \gamma_{X, Z})\circ(\gamma_{X, Y} \otimes \id_Z)$.
Morphisms of $Z(\cC)$ are morphisms of the underlying objects in $\cC$ that are compatible with the half-braiding, i.e.,
a morphism $f\colon (X,\gamma_{X,\placeholder}) \to (X',\gamma_{X',\placeholder})$ in $Z(\cC)$ is a morphism $f \colon X\to X'$ such that
\begin{align}
	\gamma_{X',Y}\circ (f\otimes \id_Y) = (\id_Y \otimes f)\circ \gamma_{X,Y}.
\end{align}
If we use string diagrams to depict calculations in tensor categories, we express the half-braiding simply as crossing of lines:
\begin{equation}
\gamma_{X, Y}
=
\tikzmath{
\draw[rounded corners] (.5, 0) -- (.5, .25) -- (0, .75) -- (0, 1);
\draw[very thick, red, mid>, rounded corners, knot] (0, 0) -- (0, .25) -- (.5, .75) -- (.5, 1);
\node[red] at (0, -.15) {$\scriptstyle X$};
\node at (.5, -.15) {$\scriptstyle Y$};
}
\end{equation}

We let $\Irr Z(\cC)$ denote a set of representative simple objects for $Z(\cC)$.
Note that there is a forgetful functor from $Z(\cC)$ to $\cC$, i.e., we may view the objects of $Z(\cC)$ as objects in $\cC$. 
In particular it makes sense to consider the space $\cC(X \to a)$ for $X\in Z(\cC)$ and $a\in \cC$.
However a simple object in $Z(\cC)$ is generally not simple in $\cC$.

We follow the conventions in \cite[\textsection 3.5]{2511.21521} in the following discussion. 
Given a labeling $\vec{a} = (a_1, \dots, a_n)$ of the marked points in $\cB$ by objects in $\Irr(\cC)$ and an  object $X \in \Irr Z(\cC)$, we let $\{w^{X, \vec{a}}_i\}$ denote an orthonormal basis of $\cC(X \to a_1 \otimes \dots \otimes a_n)$ with respect to the \emph{trace inner product}, that is,
\begin{equation}
\tikzmath{
\draw (0, -.7) -- (0, .7);
\draw[very thick, red, mid>] (0, 1) arc(180:0:.5) -- (1, -1) arc (0:-180:.5);
\roundNbox{fill=white}{(0, -.7)}{.3}{.25}{.25}{$w^{X, \vec{a}}_i$}
\roundNbox{fill=white}{(0, .7)}{.35}{.4}{.4}{$(w^{X, \vec{a}}_j)^\dag$}
\node at (-.15, 0) {$\scriptstyle \vec{a}$};
\node[red] at (-.11, -1.25) {$\scriptstyle X$};
}
=
\delta_{i,j}.
\end{equation}
By \cite[Prop.~3.4]{2511.21521}, we have that the operator
\begin{equation}
\label{eq:SpecialMinimalProjectionsTypeX}
p^{X, \vec{a}}_i
\coloneqq
\tikzmath{
\draw[very thick] (0, -.25) ellipse (1 and .25);
\draw[very thick] (1, -2.75) arc (0:-180:1 and .25);
\draw[thick, dashed] (1, -1.5) arc (0:-180:1 and .25);
\draw[thick, dashed, black!50] (-1, -1.5) arc(180:0:1 and .25);
\draw (1, -2.75) -- (1, -.25);
\draw (-1, -2.75) -- (-1, -.25);
\draw (0, -3) -- (0, -2.25);
\draw[very thick, red, mid>, knot] (0, -2.25) -- (0, -1.25);
\draw (0, -1.25) -- (0, -.5);
\roundNbox{fill=white}{(0, -2.25)}{.3}{.4}{.4}{$(w^{X, \vec{a}}_i)^\dag$}
\roundNbox{fill=white}{(0, -1.25)}{.3}{.25}{.25}{$w^{X, \vec{a}}_i$}
\node at (0.15, -2.75) {$\scriptstyle \vec{a}$};
\node at (0.15, -.7) {$\scriptstyle \vec{a}$};
}
\in \Tube_\cB
\end{equation}
is a minimal projection. 
In the above diagram, the strand labeled $\vec{a}$ denotes $a_1 \otimes \dots \otimes a_n$, and the red oriented strand denotes $X$. 
Furthermore, we adopt the convention that 
\begin{equation}
\tikzmath{
\draw[dashed] (0, 0) -- (0, 1);
}\,
=
\frac{1}{D_\cC^2} \sum_{c \in \Irr(\cC)} d_c \,
\tikzmath{
\draw (0, 0) -- (0, 1);
\node at (.15, .5) {$\scriptstyle c$};
}.
\end{equation}
Since the red strand corresponds to the object $X \in Z(\cC)$, the crossing in \eqref{eq:SpecialMinimalProjectionsTypeX} represents the half-braiding associated to $X$. 
Additionally, we have that the projection 
\begin{equation}
P^X
\coloneqq
\sum_{\vec{a}, i} p^{X, \vec{a}}_i
\end{equation}
is a minimal central projection in $\Tube_\cB$. 
One can obtain any minimal projection $p \leq P^X$ by replacing $w^{X, \vec{a}}_i$ with an appropriate unit vector $w \in \cC(X \to \chi^{\otimes n})$ in \eqref{eq:SpecialMinimalProjectionsTypeX}.

Now, let $R \subseteq \Gamma$ be a finite, connected region with boundary components $\cB_1, \dots, \cB_n$. 
By \cite[Prop.~3.5]{2511.21521}, we have that $\cS(R)$ (which is isomorphic to $p_R \cH_R$) is isomorphic to the space
\begin{equation}
\label{eq:WhatIsTheSkeinModule}
\bigoplus_{X_1, \dots, X_n \in \Irr Z(\cC)} Z(\cC)(\mathds{1} \to X_1 \otimes \dots \otimes X_n) \otimes \cC(X_1 \to \chi^{\otimes \#\cB_1}) \otimes \dots \otimes \cC(X_n \to \chi^{\otimes \#\cB_n}),
\end{equation}
where $\#\cB_k$ denotes the number of marked points in $\cB_k$ and, as before, $\chi = \bigoplus_{c \in \Irr(\cC)} c$.
The isomorphism $\Phi_R$ from the space in \eqref{eq:WhatIsTheSkeinModule} to $\cS(R)$ given in \cite[Prop.~3.5]{2511.21521} is defined as follows: 
For $\alpha \in Z(\cC)(\mathds{1} \to X_1 \otimes \dots \otimes X_n)$ and $w_k \in \cC(X_k \to a_k^{(1)} \otimes \dots \otimes a_k^{(\# \cB_k)})$ for $\vec{a}_k \coloneqq (a_k^{(1)}, \dots, a_k^{(\# \cB_k)}) \in \Irr(\cC)^{\# \cB_k}$, we have that
\begin{equation}
\label{eq:IsomorphismToSkeinModule}
\Phi_R(\alpha \otimes w_1 \otimes \dots \otimes w_n)
=
\left( \prod_{k = 1}^n d_{X_k} \right)^{1/2} D_\cC^{n - 1}
\tikzmath{
\filldraw[gray!25] (-5, -1) rectangle (4, 2.75);
\draw[thick, dashed, rounded corners] (-3.5, 0) -- (-3.5, -.5) -- (3.5, -.5) -- (3.5, 2.25) -- (-3.5, 2.25) -- (-3.5, 0);
\draw[very thick, red, mid>, gray knot] (0, 0) -- (-4, 0);
\node[red] at (-1.6, -.2) {$\scriptstyle X_1$};
\draw (-4, 0) -- (-5, 0);
\node at (-4.65, -.15) {$\scriptstyle \vec{a}_1$};
\roundNbox{fill=yellow!25}{(-4, 0)}{.3}{0}{0}{$w_1$};
\draw[thick, dashed] (-2.25, 1.25) ellipse (1 and .5);
\filldraw[fill=white] (-3, 1) rectangle (-2.5, 1.5);
\draw[very thick, red, mid>, gray knot] (-.4, 0) -- (-.4, .75) arc(0:90:.5) -- (-1.75, 1.25);
\node[red] at (-.9, 1.05) {$\scriptstyle X_2$};
\draw (-1.75, 1.25) -- (-2.5, 1.25);
\node at (-2.25, 1.1) {$\scriptstyle \vec{a}_2$};
\roundNbox{fill=yellow!25}{(-1.75, 1.25)}{.3}{0}{0}{$w_2$};
\node[red] at (.025, .875) {$\dots$};
\draw[thick, dashed] (2.25, 1.25) ellipse (1 and .5);
\filldraw[fill=white] (3, 1) rectangle (2.5, 1.5);
\draw[very thick, red, mid>, gray knot] (.4, 0) -- (.4, .75) arc(0:-90:-.5) -- (1.75, 1.25);
\node[red] at (.9, 1.05) {$\scriptstyle X_n$};
\draw (1.75, 1.25) -- (2.5, 1.25);
\node at (2.27, 1.1) {$\scriptstyle \vec{a}_n$};
\roundNbox{fill=yellow!25}{(1.75, 1.25)}{.3}{0}{0}{$w_n$};
\roundNbox{red, fill=red!10}{(0, 0)}{.3}{.3}{.3}{$\alpha$}
},
\end{equation}
where $R$ corresponds to the grey region on the right-hand side.
As before, the crossing denotes the half-braiding for $X \in Z(\cC)$, and $D_\cC$ denotes the global dimension of $\cC$.
The projection $\mathfrak{t}_{\cB_k}(P^{X_k})$ acting on $p_R\cH_R \cong \cS(R)$ projects onto the summand corresponding to $X_k$ in \eqref{eq:WhatIsTheSkeinModule}. 
In addition, for a minimal projection $p \leq P^{X_k}$ corresponding to a unit vector $w \in \cC(X_k \to \chi^{\otimes \# \cB_k})$, the projection $\mathfrak{t}_{\cB_k}(p)$ further projects onto the 1-dimensional subspace of $\cC(X_k \to \chi^{\otimes \# \cB_k})$ spanned by $w$. 

We will often considered regions $R$ that are disks with $m$ punctures, where each puncture corresponds to removing an edge $e$ and its two adjacent faces. 
In this case, we will use $\mathfrak{t}_e$ to denote the action of the tube algebra $\Tube_{\cB_e}$ along the boundary component at $e$, in a slight abuse of notation. 
Furthermore, for a region $R \subseteq \Gamma$, we will write $R \setminus e$ to denote the region with the edge $e$ and its two adjacent faces removed. 

\subsection{Drinfeld insertions and hopping operators}
\label{subsec:DrinfeldInsertions}
We now define the operators we will use in defining the anyon sectors \cite[\textsection 6]{2511.21521}.
We continue to follow the conventions in \cite{2511.21521}. 
For each $X \in \Irr Z(\cC)$, we will fix a minimal projection $p^X \leq P^X \in \Tube_{\cB_e}$, corresponding to the unit vector $w^X_e \in \cC(X \to \chi \otimes \chi)$. 
In the case where $X = \mathds{1}$, we specifically take $p^{\mathds{1}}$ to correspond to the unit vector 
\begin{equation}
w^{\mathds{1}}_e
\coloneqq
\frac{1}{D_\cC}
\sum_{c \in \Irr(\cC)}
d_c 
\tikzmath{
\draw[mid<] (0, .5) arc(-180:0:.5);
\node at (0, .65) {$\scriptstyle \bar c$};
\node at (1, .65) {$\scriptstyle c$};
}
\in \cC(\mathds{1} \to \chi \otimes \chi),
\end{equation}
where $\bar c \in \Irr(\cC)$ is the object isomorphic to the dual $c^\vee$ of $c$. 
This choice ensures that $\mathfrak{t}_e(p^{\mathds{1}}) = B_{f_1} B_{f_2}$, where $f_1$ and $f_2$ are the two faces adjacent to $e$ \cite[Lem.~5.4]{2511.21521}. 

Following \cite[Convention 6.1]{2511.21521}, we will denote
\begin{equation}
\tikzmath{
\filldraw[gray!25] (0, 0) rectangle (4, 1.5);
\filldraw[fill=white] (.5, .5) rectangle (1.5, 1);
\draw[dotted] (1, .5) -- (1, 1);
\node at (1.15, .75) {$\scriptstyle e$};
\draw[very thick, red, mid>] (4, .75) -- (1.5, .75);
\node[red] at (3.2, .95) {$\scriptstyle X$};
}
\coloneqq
\tikzmath{
\filldraw[gray!25] (0, 0) rectangle (4, 1.5);
\filldraw[fill=white] (.5, .5) rectangle (1.5, 1);
\draw[dotted] (1, .5) -- (1, 1);
\node at (1.15, .75) {$\scriptstyle e$};
\draw[thick, dashed] (1.6, .75) ellipse (1.4 and .6);
\draw[very thick, red, mid>, gray knot] (4, .75) -- (2.5, .75);
\draw (2, .75) -- (1.5, .75);
\roundNbox{fill=yellow!25}{(2.25, .75)}{.3}{.1}{.1}{$w^X_e$}
\node[red] at (3.6, .95) {$\scriptstyle X$};
},
\end{equation}
where we view this diagram as a piece of a larger diagram as shown in \eqref{eq:IsomorphismToSkeinModule}.
Similarly, for an arbitrary boundary component $\cB$ with $n$ marked points and an arbitrary unit vector $w \in \cC(X \to \chi^{\otimes n})$, we define 
\begin{equation}
\tikzmath{
\filldraw[gray!25] (0, 0) rectangle (4, 1.5);
\filldraw[yellow!25](1.5, .625) rectangle (1.75, .875);
\filldraw[fill=white] (.5, .5) rectangle (1.5, 1);
\draw[very thick, red, mid>] (4, .75) -- (1.75, .75);
\node[red] at (3.2, .95) {$\scriptstyle X$};
\node at (1.625, .75) {\tiny $w$};
}
\coloneqq
\tikzmath{
\filldraw[gray!25] (0, 0) rectangle (4, 1.5);
\filldraw[fill=white] (.5, .5) rectangle (1.5, 1);
\draw[thick, dashed] (1.6, .75) ellipse (1.4 and .6);
\draw[very thick, red, mid>, gray knot] (4, .75) -- (2.5, .75);
\draw (2, .75) -- (1.5, .75);
\roundNbox{fill=yellow!25}{(2.25, .75)}{.25}{0}{0}{$w$}
\node[red] at (3.6, .95) {$\scriptstyle X$};
}.
\end{equation}
The boundary component $\cB$ in the above equation will usually be an external boundary component, so the picture will more closely resemble the external boundary component of \eqref{eq:IsomorphismToSkeinModule}. 

We now write down the \emph{Drinfeld insertion operators} defined in \cite[\textsection 6.1]{2511.21521}. 
Let $R \subseteq \Gamma$ be a finite connected region that is a disk with $n$ punctures, and suppose further that each puncture corresponds to removing an edge $e_n$ and its two adjacent faces. 
Let $X_1, \dots, X_n, Y_1, \dots, Y_n \in \Irr Z(\cC)$, and let $\beta \colon X_1 \otimes \dots \otimes X_n \to Y_1 \otimes \dots \otimes Y_n$ be a morphism in $Z(\cC)$. 
We define an operator $\Dr_R(\beta) \in B(\cH_R)$ that only acts nontrivially on $\prod \mathfrak{t}_{e_k}(p^{X_k}) p_R \cH_R$. 
Note that by \cite[Prop.~4.7]{2511.21521}, the subspace $\prod \mathfrak{t}_{e_k}(p^{X_k}) p_R \cH_R$ is isomorphic to the space spanned by diagrams of the form 
\begin{equation}
\tikzmath{
\filldraw[gray!25] (-3, -1) rectangle (3, 2.75);
\draw[very thick, red, mid>] (0, 0) -- (-3, 0);
\node[red] at (-1.6, -.2) {$\scriptstyle X_0$};
\filldraw[yellow!25] (-3, -.125) rectangle (-2.75, .125);
\node at (-2.875, 0) {\tiny $w$};
\filldraw[fill=white] (-2.5, 1) rectangle (-1.5, 1.5);
\draw[dotted] (-2, 1) -- (-2, 1.5);
\node at (-1.8, 1.25) {$\scriptstyle e_1$};
\draw[very thick, red, mid>] (-.4, 0) -- (-.4, .75) arc(0:90:.5) -- (-1.5, 1.25);
\node[red] at (-1.1, 1.05) {$\scriptstyle X_1$};
\filldraw[fill=white] (-1.7, 1.75) rectangle (-.7, 2.25);
\draw[dotted] (-1.2, 1.75) -- (-1.2, 2.25);
\draw[very thick, red, mid>] (-.2, 0) -- (-.2, 1.5) arc(0:90:.5);
\node[red] at (-.45, 1.45) {$\scriptstyle X_2$};
\node[red] at (.025, 1.875) {$\dots$};
\filldraw[fill=white] (1.7, 1.75) rectangle (.7, 2.25);
\draw[dotted] (1.2, 1.75) -- (1.2, 2.25);
\draw[very thick, red, mid>] (.2, 0) -- (.2, 1.5) arc(0:-90:-.5);
\node[red] at (.65, 1.45) {$\scriptstyle X_{n - 1}$};
\filldraw[fill=white] (2.5, 1) rectangle (1.5, 1.5);
\draw[dotted] (2, 1) -- (2, 1.5);
\node at (2.2, 1.25) {$\scriptstyle e_n$};
\draw[very thick, red, mid>] (.4, 0) -- (.4, .75) arc(0:-90:-.5) -- (1.5, 1.25);
\node[red] at (1.1, 1.05) {$\scriptstyle X_n$};
\roundNbox{red, fill=red!10}{(0, 0)}{.3}{.3}{.3}{$\alpha$}
}
\end{equation}
The map $\Dr_R(\beta)$ is then defined as follows: 
\begin{equation}
\label{eq:DrinfeldInsertions}
\tikzmath{
\filldraw[gray!25] (-3, -1) rectangle (3, 2.75);
\draw[very thick, red, mid>] (0, 0) -- (-3, 0);
\node[red] at (-1.6, -.2) {$\scriptstyle X_0$};
\filldraw[yellow!25] (-3, -.125) rectangle (-2.75, .125);
\node at (-2.875, 0) {\tiny $w$};
\filldraw[fill=white] (-2.5, 1) rectangle (-1.5, 1.5);
\draw[dotted] (-2, 1) -- (-2, 1.5);
\node at (-1.8, 1.25) {$\scriptstyle e_1$};
\draw[very thick, red, mid>] (-.4, 0) -- (-.4, .75) arc(0:90:.5) -- (-1.5, 1.25);
\node[red] at (-1.1, 1.05) {$\scriptstyle X_1$};
\filldraw[fill=white] (-1.7, 1.75) rectangle (-.7, 2.25);
\draw[dotted] (-1.2, 1.75) -- (-1.2, 2.25);
\draw[very thick, red, mid>] (-.2, 0) -- (-.2, 1.5) arc(0:90:.5);
\node[red] at (-.45, 1.45) {$\scriptstyle X_2$};
\node[red] at (.025, 1.875) {$\dots$};
\filldraw[fill=white] (1.7, 1.75) rectangle (.7, 2.25);
\draw[dotted] (1.2, 1.75) -- (1.2, 2.25);
\draw[very thick, red, mid>] (.2, 0) -- (.2, 1.5) arc(0:-90:-.5);
\node[red] at (.65, 1.45) {$\scriptstyle X_{n - 1}$};
\filldraw[fill=white] (2.5, 1) rectangle (1.5, 1.5);
\draw[dotted] (2, 1) -- (2, 1.5);
\node at (2.2, 1.25) {$\scriptstyle e_n$};
\draw[very thick, red, mid>] (.4, 0) -- (.4, .75) arc(0:-90:-.5) -- (1.5, 1.25);
\node[red] at (1.1, 1.05) {$\scriptstyle X_n$};
\roundNbox{red, fill=red!10}{(0, 0)}{.3}{.3}{.3}{$\alpha$}
}
\mapsto
\left( \prod_{k = 1}^n \frac{d_{Y_k}}{d_{X_k}}\right)^{1/2}
\tikzmath{
\filldraw[gray!25] (-3, -2.5) rectangle (3, 2.75);
\draw[very thick, red, mid>] (0, -1.5) -- (-3, -1.5);
\node[red] at (-1.6, -1.7) {$\scriptstyle X_0$};
\filldraw[yellow!25] (-3, -1.625) rectangle (-2.75, -1.375);
\node at (-2.875, -1.5) {\tiny $w$};
\filldraw[fill=white] (-2.5, 1) rectangle (-1.5, 1.5);
\draw[dotted] (-2, 1) -- (-2, 1.5);
\node at (-1.8, 1.25) {$\scriptstyle e_1$};
\draw[very thick, red, mid>] (-.4, 0) -- (-.4, .75) arc(0:90:.5) -- (-1.5, 1.25);
\node[red] at (-1.1, 1.05) {$\scriptstyle Y_1$};
\draw[very thick, red, mid>] (-.4, -1.5) -- (-.4, 0);
\node[red] at (-.6, -.6) {$\scriptstyle X_1$};
\filldraw[fill=white] (-1.7, 1.75) rectangle (-.7, 2.25);
\draw[dotted] (-1.2, 1.75) -- (-1.2, 2.25);
\draw[very thick, red, mid>] (-.2, 0) -- (-.2, 1.5) arc(0:90:.5);
\node[red] at (-.45, 1.45) {$\scriptstyle Y_2$};
\draw[very thick, red, mid>] (-.2, -1.5) -- (-.2, 0);
\node[red] at (.025, 1.875) {$\dots$};
\node[red] at (0.025, -.6) {$\scriptstyle \dots$};
\filldraw[fill=white] (1.7, 1.75) rectangle (.7, 2.25);
\draw[dotted] (1.2, 1.75) -- (1.2, 2.25);
\draw[very thick, red, mid>] (.2, 0) -- (.2, 1.5) arc(0:-90:-.5);
\node[red] at (.65, 1.45) {$\scriptstyle Y_{n - 1}$};
\draw[very thick, red, mid>] (.2, -1.5) -- (.2, 0);
\filldraw[fill=white] (2.5, 1) rectangle (1.5, 1.5);
\draw[dotted] (2, 1) -- (2, 1.5);
\node at (2.2, 1.25) {$\scriptstyle e_n$};
\draw[very thick, red, mid>] (.4, 0) -- (.4, .75) arc(0:-90:-.5) -- (1.5, 1.25);
\node[red] at (1.1, 1.05) {$\scriptstyle Y_n$};
\draw[very thick, red, mid>] (.4, -1.5) -- (.4, 0);
\node[red] at (.65, -.6) {$\scriptstyle X_n$};
\roundNbox{red, fill=red!10}{(0, -1.5)}{.3}{.3}{.3}{$\alpha$}
\roundNbox{red, fill=red!10}{(0, 0)}{.3}{.3}{.3}{$\beta$}
}
\end{equation}
The normalization present ensures that $\Dr_R(\beta)^* = \Dr_R(\beta^\dag)$. 

We will often consider Drinfeld insertion operators along links. 
Following \cite[\textsection 6.3]{2511.21521}, a \emph{link} is a finite dual path (i.e., a path of neighbouring faces) $L = (f_0, \dots, f_n)$ satisfying the following conditions: 
\begin{itemize}
\item
$f_i$ and $f_j$ share an edge if and only if $i = j\pm 1$, and 
\item
the first four faces of $L$ form a straight line, as do the last four faces of $L$.
\end{itemize}
Note that the first condition ensures that $L$ is self-avoiding. 
We define $\partial_iL$ to be the edge between $f_0$ and $f_1$, and we define $\partial_fL$ to be the edge between $f_{n - 1}$ and $f_n$. 
To a link $L$, we associate a region $R^L$ that consists of the union of all edges and faces comprising $L$, \emph{except} for the faces $f_0, f_1, f_{n - 1}, f_n$ and the edges $\partial_iL, \partial_fL$. 
Note that $R^L$ is homeomorphic to a disk with two punctures, and the two punctures both correspond to an edge and its two adjacent faces, see the following illustration: 
\begin{equation}
\tikzmath{
\filldraw[gray!25] (-.2, -.2) rectangle (6.2, 1.2);
\filldraw[white] (.2, .2) rectangle (1.8, .8);
\filldraw[white] (4.2, .2) rectangle (5.8, .8);
\draw (0, 0) rectangle (6, 1);
\draw[dotted] (1, 0) -- (1, 1);
\draw[dotted] (5, 0) -- (5, 1);
\foreach \n in {2, ..., 4}
{
\draw (\n, 0) -- (\n, 1);
}
\node at (2.5, .5) {$\scriptstyle L$};
\node[gray] at (3.5, .5) {$\scriptstyle R^L$};
}
\end{equation}

Finally, we define Drinfeld insertion operators that create/remove anyon pairs and move anyons from one puncture to another. 
Let $X \in \Irr Z(\cC)$, and let $\bar X \in \Irr Z(\cC)$ be the representative corresponding to $X^\vee$, the dual of $X$. 
We then have the following operators \cite[Eq.~49-51]{2511.21521}, 
\begin{gather}
\Dr_L^{(\mathds{1} \mathds{1}\to X \bar X)}
\coloneqq
d_X^{-1/2} \Dr_{R^L} \left[
\tikzmath{
\draw[mid<, very thick, red] (0, .75) arc(-180:0:.25);
\draw[thick, dotted] (0, 0) -- (0, .25);
\draw[thick, dotted] (.5, 0) -- (.5, .25);
\node[red] at (0, .9) {$\scriptstyle X$};
\node[red] at (.5, .9) {$\scriptstyle \bar X$};
} \right],
\qquad
\Dr_L^{(X \bar X \to \mathds{1} \mathds{1})}
\coloneqq
d_X^{-1/2} \Dr_{R^L} \left[
\tikzmath{
\draw[mid>, very thick, red] (0, 0) arc(180:0:.25);
\draw[thick, dotted] (0, .5) -- (0, .75);
\draw[thick, dotted] (.5, .5) -- (.5, .75);
\node[red] at (0, -.15) {$\scriptstyle X$};
\node[red] at (.5, -.15) {$\scriptstyle \bar X$};
} \right],
\\
\Dr_L^{(X\mathds{1} \to \mathds{1} X)}
\coloneqq
\Dr_{R^L} \left[ 
\tikzmath{
\draw[very thick, red, rounded corners, mid>] (0, 0) -- (0, .25) -- (.5, .5) -- (.5, .75);
\draw[thick, dotted] (0, .5) -- (0, .75);
\draw[thick, dotted] (.5, 0) -- (.5, .25);
} \right],
\qquad
\Dr_L^{(\mathds{1}X \to X\mathds{1})}
\coloneqq
\Dr_{R^L} \left[ 
\tikzmath{
\draw[very thick, red, rounded corners, mid>] (.5, 0) -- (.5, .25) -- (0, .5) -- (0, .75);
\draw[thick, dotted] (.5, .5) -- (.5, .75);
\draw[thick, dotted] (0, 0) -- (0, .25);
} \right],
\end{gather}
where we omit the (non-canonical) fixed unitary $X^\vee \to \bar X$. 
Note that $(\Dr_L^{(\mathds{1} \mathds{1} \to X \bar X)})^* = \Dr_L^{(X \bar X \to \mathds{1} \mathds{1})}$ and $(\Dr_L^{(X \mathds{1} \to \mathds{1} X)})^* = \Dr_L^{(\mathds{1} X \to X \mathds{1} )}$. 
By \cite[Lem.~6.8]{2511.21521}, these four operators are partial isometries with the following source and range projections: 
\begin{align}
P_L^{\mathds{1} \mathds{1}}
&\coloneqq
\mathfrak{t}_1^L(p^{\mathds{1}}) \mathfrak{t}_2^L(p^{\mathds{1}}) p_{R^L}
=
\Dr_L^{(X \bar X \to \mathds{1} \mathds{1})} \Dr_L^{(\mathds{1} \mathds{1} \to X \bar X)}
\\
\label{eq:ProjectionForAnyonsAnnihilating}
P_L^{X \bar X}
&\coloneqq
\mathfrak{t}_0^L(P^\mathds{1}) \mathfrak{t}_1(p^{X}) \mathfrak{t}_2(p^{\bar X}) p_{R^L}
=
\Dr_L^{(\mathds{1} \mathds{1} \to X \bar X)} \Dr_L^{(X \bar X \to \mathds{1} \mathds{1})}
\\
P_L^{X \mathds{1}}
&\coloneqq
\mathfrak{t}_1^L(p^{X}) \mathfrak{t}_2^L(p^{\mathds{1}}) p_{R^L}
=
\Dr_L^{(\mathds{1} X \to X \mathds{1} )} \Dr_L^{(X \mathds{1} \to \mathds{1} X)}
\\
P_L^{\mathds{1} X}
&\coloneqq
\mathfrak{t}_1^L(p^{\mathds{1}}) \mathfrak{t}_2^L(p^{X}) p_{R^L}
=
\Dr_L^{(X \mathds{1} \to \mathds{1} X)} \Dr_L^{(\mathds{1} X \to X \mathds{1} )}
\end{align}
Here $\mathfrak{t}_1^L$ and $\mathfrak{t}_2^L$ denote the tube algebra action at $\partial_i L$ and $\partial_f L$ respectively, and $\mathfrak{t}_0^L$ denotes the tube algebra action along the external boundary component of $R^L$. 

Note that the four projections $P_L^{\mathds{1} \mathds{1}}, P_L^{X \bar X}, P_L^{X \mathds{1}}, P_L^{\mathds{1} X}$ are all disjoint. 
Following \cite{2511.21521}, we define $P_L^{??} \coloneqq P_L^{\mathds{1} \mathds{1}} + P_L^{X \bar X} + P_L^{X \mathds{1}} + P_L^{\mathds{1} X}$.
We then define the following unitary along $L$ \cite[Eq.~52]{2511.21521}:
\begin{equation}
\label{eq:UnitaryChargeTransporter}
u^X_L
\coloneqq
\Dr_L^{(\mathds{1} \mathds{1} \to X \bar X)} + \Dr_L^{(X \bar X \to \mathds{1} \mathds{1})} + \Dr_L^{(X \mathds{1} \to \mathds{1} X)} + \Dr_L^{(\mathds{1} X \to X \mathds{1} )} + (\mathds{1} - P_L^{??}).
\end{equation}

\subsection{Anyon sectors}
\label{subsec:LWAnyonSectors}
We define a chain $\scrC$ to be a sequence of links $(L_n)_{n \geq 1}$ satisfying that $\partial_iL_n = \partial_f L_{n + 1}$ and that for every $m \leq n$, the links $L_m, \dots, L_n$ can be concatenated to form a well-defined link $L_{n \to m}$. 
We fix $X \in \Irr Z(\cC)$. 
For $m \leq n$, we define $U_{\scrC[n \to m]}^X \coloneqq u_{L_n}^X \cdots u_{L_m}^X$, and we define the endomorphism $\rho^X_{\scrC} \colon \fA \to \fA$ by 
\begin{equation}
\rho^X_{\scrC}(x)
=
\lim_{n \to \infty} (U_{\scrC[n \to 1]}^X)^* x U_{\scrC[n \to 1]}^X.
\end{equation}
Note that this limit is well-defined and gives a unital $*$-endomorphism of $\fA$ \cite[Lem.~6.10]{2511.21521}.
Let $e_0 \coloneqq \partial_f L_1$. 
We let $\omega^X$ denote the unique (pure) state satisfying the following \cite[Prop.~5.2]{2511.21521} 
\begin{equation}\label{eq:anyonstate}
\omega^X(\mathfrak{t}_{e_0}(p^{\bar X})) = 1,
\qquad
\omega^X(A_e) = 1
\quad
\forall e \neq e_0,
\qquad
\omega^X(B_f) = 1 
\quad
\forall f \text{ such that } e \notin f.
\end{equation}
(Note that the first constraint uses $p^{\bar X}$, not $p^X$!)
Recalling that $(\pi_0, \cH, \Omega)$ is the GNS representation for $\omega_0$, the unique frustration-free ground state for the Levin--Wen Hamiltonian, we have that $(\pi_0 \circ \rho^X_{\scrC}, \cH, \Omega)$ is a GNS representation for $\omega^X$ \cite[Prop.~7.2]{2511.21521}. 
In particular, we have that 
\begin{equation}
\label{eq:AnyonStateIsAnyonState}
\omega^X(-) = \bra{\Omega} \pi_0 \circ \rho^X_{\scrC}(-) \ket{\Omega}.
\end{equation}
Furthermore, $\pi_0 \circ \rho^X_{\scrC}$ is an anyon sector of type $X$ \cite[Prop.~7.10]{2511.21521}, and if $\scrC$ is contained in the cone $\Lambda$, that is, every link in $\scrC$ is contained in $\Lambda$, then $\pi_0 \circ \rho^X_{\scrC}$ is localized in $\Lambda$. 

Since $\fA$ is a UHF algebra, it is simple, so $\pi_0 \colon \fA \to B(\cH)$ is faithful. 
We may therefore identify $\pi_0(\fA)$ with $\fA$. 
In particular, we will identify $\pi_0 \circ \rho^X_{\scrC}$ with $\rho^X_{\scrC}$. 

\subsection{Rigorous interpretation of heuristic}
\label{subsec:HeuristicMadeRigorous}
We are finally able to provide a rigorous interpretation of the physical argument in the main text. 
We fix a nonabelian anyon type $X \in \Irr Z(\cC)$. 
We consider a chain $\scrC$ that is contained in $(\Lambda^{+1})^c$, where $\Lambda^{+1}$ is the region consisting of all points distance at most $1$ from $\Lambda$. 
To ease the notation, we will write $\rho^X$ for $\rho^X_{\scrC}$. 
For each link $L_n$ we let $e_n \coloneqq \partial_i L_n$ and $e_{n - 1} \coloneqq \partial_f L_n$. 

Consider the first link $L_1$ of $\scrC$, and consider a site $s$ in the bulk of $\Lambda$ that consists of an edge and its two adjacent faces. 
We require that $s \in ((\Lambda^c)^{+1})^c$.
Note that $s$ is disjoint from the link $L_1$. 
Let $R$ be a thrice-punctured disk, whose first puncture is at $\partial_i L_1$, whose second puncture is at $\partial_f L_1$, and whose third puncture is at $s$. 
We consider an isometry $\varsigma_a \colon \mathds{1} \to X \otimes \bar X \otimes \bar a$, where $a \leq X \otimes \bar X$ and $a \neq \mathds{1}$. 
Such an $a$ exists since $X$ is a nonabelian anyon. 

In this section, we prove the following theorem. 

\begin{thm}
\label{thm:LevinWenViolatesUoP}
The states 
\begin{equation}
\omega^X(-) = \bra{\Omega} \rho^X(-) \ket{\Omega},
\qquad\qquad
\omega^X_a(-) \coloneqq  \bra{\Omega} \Dr_{R}(\varsigma_a)^\dag \rho^X(-) \Dr_{R}(\varsigma_a) \ket{\Omega}
\end{equation}
violate the uniqueness of purifications criterion in \cite{MR5026085} with respect to the cone $\Lambda$. 
In more detail, the following hold:
\begin{itemize}
\item
	For all $x \in \rho^X_{\Lambda^c}(\cR(\Lambda^c))$, we have 
\begin{equation}
\bra{\Omega} x \ket{\Omega}
=
\bra{\Omega} \Dr_{R}(\varsigma_a)^\dag x \Dr_{R}(\varsigma_a) \ket{\Omega},
\end{equation}
\item
	For any $u \in \rho^X_\Lambda(\cR(\Lambda)) = \cR(\Lambda)$, we have 
\begin{equation}
	\bra{\Omega} u \Dr_{R}(\varsigma_a) \ket{\Omega} = 0.
\end{equation}
\end{itemize}
\end{thm}

We prove this theorem in parts. 
First, we characterize the state $\omega^X_a$ in Lemma \ref{lem:SplitNonAbelianAnyonStateProjections}.
(Note that the state $\omega^X$ is already characterized by \cref{eq:anyonstate}.)
We then show in Lemma \ref{lem:NonAbelianLTQO} that the state $\omega^X$ satisfies an LTQO property \cite{MR3077916}, which was implicitly proven in \cite{2511.21521}. 
After these preliminaries, we prove the two parts of Theorem \ref{thm:LevinWenViolatesUoP}.
In Lemma \ref{lem:TwoStatesAgreeOutsideCone}, we show that the two states agree on $\rho^X_\Lambda(\cR(\Lambda^c))$, and in Lemma \ref{lem:SplitAnyonCannotBeRemovedByUnitary} we show that the states cannot be connected by a unitary in $\rho^X_\Lambda(\cR(\Lambda)) = \cR(\Lambda)$.

\begin{lem}
\label{lem:SplitNonAbelianAnyonStateProjections}
We have that
\begin{equation}
\omega^X_a(\mathfrak{t}_{s}(p^{\bar a}))
=
\omega^X_a(\mathfrak{t}_{e_0}(p^{\bar X}))
=
\omega^X_a(A_e)
=
\omega^X_a(B_f)
=
1,
\end{equation}
where $e$ and $f$ range over all edges and faces that are disjoint from the punctures at $s$ and $e_0$. 
\end{lem}

In the proof of this lemma, we will identify vectors in $\cH$ with vectors in the ``skein module," depicted graphically. 
These statements do not type check, since the skein module is only defined for finite regions.
Nonetheless, these pictures act as a useful bookkeeping device.

\begin{proof}[Proof of Lemma \ref{lem:SplitNonAbelianAnyonStateProjections}]
We first observe that $\Dr_R(\varsigma_a) \ket{\Omega}$ corresponds to the following vector in the skein module:
\begin{equation}
d_a^{1/2} d_X
\tikzmath{
\filldraw[gray!25](-5, -.5) rectangle (6, 1);
\filldraw[fill=white] (-3.75, 0) rectangle (-2.75, .5);
\draw[dotted] (-3.25, 0) -- (-3.25, .5);
\node at (-3.1, .25) {\tiny $s$};
\draw[mid>, very thick, blue, rounded corners] (-2.75, .25) -- (0, -.25) -- (1.625, -.25) -- (1.625, .25);
\node[blue] at (-1.95, -.05) {$\scriptstyle a$};
\filldraw[fill=white] (.25, 0) rectangle (1.25, .5);
\draw[dotted] (.75, 0) -- (.75, .5);
\node at (.91, .25) {\tiny $e_0$};
\draw[mid>, very thick, red] (1.25, .25) -- (1.75, .25);
\node[red] at (1.5, .5) {$\scriptstyle X$};
\filldraw[fill=white] (1.75, 0) rectangle (2.75, .5);
\draw[dotted] (2.25, 0) -- (2.25, .5);
\node at (2.41, .25) {\tiny $ e_{1}$};
\filldraw[cyan] (1.625, .25) circle (0.05cm);
\draw[thick, dashed] (-5, 1.25) -- (-3.25, -1) -- (-1.5, 1.25);
\node at (-3.25, 1.25) {$\scriptstyle \Lambda$};
}.
\end{equation}
	In the above picture, the cyan dot represents $\varsigma_a$.
In particular, we have that 
\begin{equation}
\Dr_R(\varsigma_a)\ket{\Omega} 
\in 
\mathfrak{t}_{s}(p^{\bar a}) \mathfrak{t}_{e_0}(p^{\bar X}) \mathfrak{t}_{e_1}(p^X)p_{S} \cH,
\end{equation}
where $S$ is any region with punctures at $s$, $e_0$, and $e_1$ such that $R \subseteq S$. 

Now, recall that $\rho^X(\cdot) = \lim_{n \to \infty} (U_n^X)^* \cdot U_n^X$, where $U_n^X = u^X_{L_n} \cdots u^X_{L_1}$. 
Since $s$ is disjoint from $L_1$, we have that 
\begin{equation}
\Dr_R(\varsigma_a)\ket{\Omega} 
\in 
\mathfrak{t}_{\partial_{\mathrm{ext}}L_1}(P^{\bar a}) \mathfrak{t}_{e_0}(p^{\bar X}) \mathfrak{t}_{e_1}(p^X)p_{R^{L_1}} \cH,
\end{equation}
where $\partial_{\mathrm{ext}}L_1$ denotes the exterior boundary of $L_1$. 
Therefore, we have that $\Dr_R(\varsigma_a)\ket{\Omega} \in P^{??}_{L_1}\cH^\perp$, so we have that $u_{L_1} \Dr_R(\varsigma_a)\ket{\Omega} = \Dr_R(\varsigma_a)\ket{\Omega}$. 
Now, letting $L_{n \to 2}$ be the concatenation of the links $L_2, \dots, L_n$, we observe that $\Dr_R(\varsigma_a)\ket{\Omega} \in P^{\mathds{1} X}_{L_{n \to 2}} \cH$, so by \cite[Lem.~6.9]{2511.21521}, we have that 
\begin{align}
U_n^X\Dr_R(\varsigma_a)\ket{\Omega} 
&=
u_{L_n}^X \cdots u_{L_2}^X u_{L_1}^X \Dr_R(\varsigma_a)\ket{\Omega} 
=
u_{L_n}^X \cdots u_{L_2}^X \Dr_R(\varsigma_a)\ket{\Omega} 
\\&=
u_{L_{n \to 2}}^X \Dr_R(\varsigma_a)\ket{\Omega}
=
\Dr_{L_{n \to 2}}^{(\mathds{1} X \to X \mathds{1})} \Dr_R(\varsigma_a)\ket{\Omega}.
\end{align}
The last step in the above equation follows from the definition of $u_{L_{n \to 2}}^X$ since $\Dr_R(\varsigma_a)\ket{\Omega} \in P^{\mathds{1} X}_{L_{n \to 2}} \cH$.
Therefore, $U_n^X\Dr_R(\varsigma_a)\ket{\Omega}$ corresponds to the following vector in the skein module: 
\begin{equation}
\label{eq:StateWithXOutsideCone+aInsideCone-nthStage}
d_a^{1/2} d_X
\tikzmath{
\filldraw[gray!25](-5, -.5) rectangle (6, 1);
\filldraw[fill=white] (-3.75, 0) rectangle (-2.75, .5);
\draw[dotted] (-3.25, 0) -- (-3.25, .5);
\node at (-3.1, .25) {\tiny $s$};
\draw[mid>, very thick, blue, rounded corners] (-2.75, .25) -- (0, -.25) -- (3.6875, -.25) -- (3.6875, .25);
\node[blue] at (-1.95, -.05) {$\scriptstyle a$};
\filldraw[fill=white] (.25, 0) rectangle (1.25, .5);
\draw[dotted] (.75, 0) -- (.75, .5);
\node at (.91, .25) {\tiny $e_0$};
\draw[mid>, very thick, red] (1.25, .25) -- (4.5, .25);
\node[red] at (2.25, .45) {$\scriptstyle X$};
\filldraw[fill=white] (4.5, 0) rectangle (5.5, .5);
\draw[dotted] (5, 0) -- (5, .5);
\node at (5.16, .25) {\tiny $ e_{n}$};
\filldraw[cyan] (3.6875, .25) circle (0.05cm);
\draw[thick, dashed] (-5, 1.25) -- (-3.25, -1) -- (-1.5, 1.25);
\node at (-3.25, 1.25) {$\scriptstyle \Lambda$};
}.
\end{equation}
In particular, we have that 
\begin{equation}
U_n^X\Dr_R(\varsigma_a)\ket{\Omega}
\in 
\mathfrak{t}_{s}(p^{\bar a}) \mathfrak{t}_{e_0}(p^{\bar X}) \mathfrak{t}_{e_n}(p^X)p_{S} \cH,
\end{equation}
where $S$ is any thrice-punctured disk with punctures at $s$, $e_0$, and $e_n$. 

Thus, since $\omega^X_a(-) = \bra{\Omega} \Dr_{R}(\varsigma_a)^\dag \rho^X(-) \Dr_{R}(\varsigma_a) \ket{\Omega}$, we have that 
\begin{equation}
\omega^X_a(\mathfrak{t}_{s}(p^{\bar a}))
=
\omega^X_a(\mathfrak{t}_{e_0}(p^{\bar X}))
=
\omega^X_a(A_e)
=
\omega^X_a(B_f)
=
1,
\end{equation}
where $e$ and $f$ range over all edges and faces that are disjoint from the punctures at $s$ and $e_0$. 
\end{proof}

We now observe that the state $\omega^X$ satisfies an LTQO property \cite{MR3077916}. 
Let $S$ be a finite simply connected region containing the puncture at $e_0$ such that $S \setminus e_0$ is an annulus, where $S \setminus e_0$ is the region obtained by removing $e_0$ and its two adjacent faces from $S$. 
We then define $p_S^X \coloneqq \mathfrak{t}_{e_0}(p^{\bar X}) p_{S \setminus e_0}$.
\begin{lem}
\label{lem:NonAbelianLTQO}
Let $S$ be a simply connected region containing the puncture at $e_0$ such that $S \setminus e_0$ is an annulus. 
We then have that $p_{S^{+1}}^X \fA(S) p_{S^{+1}}^X = \bbC p_{S^{+1}}^X$, where $S^{+1}$ denotes all points distance at most 1 away from $S$ as before.
In particular, the LTQO condition of \cite{MR3077916} is satisfied. 
\end{lem}

\begin{proof}
By \cite[Lem.~4.9 \& 4.11]{2511.21521}, for any density matrix $\rho \in B(\cH_{S^{+1}})$ supported on $p_{S^{+1}}^X \cH_{S^{+1}}$, we have that $\Tr_{S^{+1} \setminus S}(\rho)$ does not depend on the specific density matrix $\rho$. 
Therefore, by a standard argument (see for instance \cite[Prop.~1]{2404.05867}), we have that $p_{S^{+1}}^X \fA(S) p_{S^{+1}}^X = \bbC p_{S^{+1}}^X$.
\end{proof}

\begin{lem}
\label{lem:TwoStatesAgreeOutsideCone}
For all $x \in \rho^X_\Lambda(\cR(\Lambda^c))$, we have that 
\(
\bra{\Omega} x \ket{\Omega}
=
\bra{\Omega}\Dr_{R}(\varsigma_a)^\dag x \Dr_{R}(\varsigma_a) \ket{\Omega}.
\)
\end{lem}

\begin{proof}
	Writing $x = \rho^X_\Lambda(y)$, we wish to show that $$\bra{\Omega} \rho^X_\Lambda(y) \ket{\Omega} = \bra{\Omega} \Dr_{R}(\varsigma_a)^\dag \rho^X_\Lambda(y)  \Dr_{R}(\varsigma_a) \ket{\Omega}$$ for all $y \in \cR(\Lambda^c)$. 
By normality, it suffices to show that $\omega^X(y) = \omega^X_a(y)$ for all $y \in \fA(\Lambda^c)_{\loc}$. 
Let $y \in \fA(\Lambda^c)_{\loc}$, and let $S \subseteq \Lambda^c$ be a simply connected region containing the support of $y$ and the puncture at $e_0$, such that $S \setminus e_0$ is an annulus. 
Then $y \in \fA(S)$, so $p_{S^{+1}}^X y p_{S^{+1}}^X = \lambda p_{S^{+1}}^X$ for some $\lambda \in \bbC$. 
Now, $S^{+1} \subseteq (\Lambda^c)^{+1}$, so it is disjoint from the site $s$. 
Therefore, we have that 
\begin{equation}
\omega^X_a(p_{S^{+1}}^X)
=
1
=
\omega^X(p_{S^{+1}}^X).
\end{equation}
Hence, applying the standard result found in \cite[\textsection 2.1.1]{MR2345476}, we have that 
\begin{equation}
\omega^X_a(y)
=
\omega^X_a(p_{S^{+1}}^X y p_{S^{+1}}^X)
=
\lambda \omega^X_a(p_{S^{+1}}^X)
=
\lambda
=
\lambda \omega^X(p_{S^{+1}}^X)
=
\omega^X(p_{S^{+1}}^X y p_{S^{+1}}^X)
=
\omega^X(y),
\end{equation}
as desired.
\end{proof}

\begin{lem}
\label{lem:SplitAnyonCannotBeRemovedByUnitary}
For every $u \in \rho^X_\Lambda(\cR(\Lambda)) = \cR(\Lambda)$, we have that $\bra{\Omega} u \Dr_{R}(\varsigma_a) \ket{\Omega} = 0$. 
\end{lem}

\begin{proof}
By continuity, it suffices to prove the result for $u \in \fA(\Lambda)_{\loc}$. 
Let $u \in \fA(\Lambda)_{\loc}$. 
Since the chain $\scrC$ is localized in $(\Lambda^{+1})^c$ and $s \in \Lambda$, there exists a disk $S$ containing $\supp(u)^{+1}$ and the site $s$ but disjoint from the puncture at $e_0$. 
We further assume that $S \setminus s$ is an annulus, which is possible since $s \in ((\Lambda^c)^{+1})^c$. 
Since $S$ contains $s$ and is disjoint from the puncture at $e_0$, the density matrix $\rho^X_{a, S \setminus s}$ for $\omega^X_a|_{\fA(S \setminus s)}$ is supported on $\mathfrak{t}_s(p^{\bar a})p_{S \setminus s} \cH_{S \setminus s}$. 
	By \cite[Prop.~4.7]{2511.21521}, we have that $\mathfrak{t}_s(p^{\bar a})p_{S \setminus s} \cH_{S \setminus s} = \mathfrak{t}_{\partial S}(P^a) \mathfrak{t}_s(p^{\bar a})p_{S \setminus s} \cH_{S \setminus s}$. 
Indeed, by \cite[Prop.~4.7]{2511.21521}, we have that $p_{S \setminus s} \cH_{S \setminus s}$ is isomorphic as $\Tube_s$ and $\Tube_{\partial S}$ modules to the space 
\begin{gather}
\bigoplus_{X_1, X_2 \in \Irr(Z(\cC))} Z(\cC)(\mathds{1} \to X_1 \otimes X_2)  \otimes \cC(X_1 \to \chi^{\otimes 2}) \otimes \cC(X_2 \to \chi^{\#\partial S})
\\
= \bigoplus_{X_1 \in \Irr(Z(\cC))} Z(\cC)(\mathds{1} \to X_1 \otimes \overline{X_1})  \otimes \cC(X_1 \to \chi^{\otimes 2}) \otimes \cC(\overline{X_1} \to \chi^{\#\partial S}).
\end{gather}
The projection $\mathfrak{t}_s(p^{\bar a})$ projects onto a 1-dimensional subspace of $\bigoplus_{X_1 \in \Irr(Z(\cC))}\cC(X_1 \to \chi^{\otimes 2})$ for which $X_1 = \bar a$.
Therefore, $\mathfrak{t}_{\partial S}(P^{a})$, which projects onto the subspace corresponding to $X_1 = \bar a$, fixes $\mathfrak{t}_s(p^{\bar a})p_{S \setminus s} \cH_{S \setminus s}$.
	We thus have that $\omega^X_a(\mathfrak{t}_{\partial S}(P^a)) = 1$.Therefore,
\begin{equation}
\bra{\Omega} \Dr_R(\varsigma_a)^\dag \rho^X(\mathfrak{t}_{\partial S}(P^a)) \Dr_R(\varsigma_a) \ket{\Omega}
=
\omega^X_a(\mathfrak{t}_{\partial S}(P^a))
=
1,
\end{equation}
	and hence $\rho^X(\mathfrak{t}_{\partial S}(P^a)) \Dr_R(\varsigma_a) \ket{\Omega} = \Dr_R(\varsigma_a) \ket{\Omega}$. 

On the other hand, since $S$ is disjoint from the puncture at $e_0$, the density matrix $\rho^X_S$ for $\omega^X|_{\fA(S)}$ is supported on $p_S\cH_S$. 
	But $\mathfrak{t}_{\partial S}(P^a)p_S\cH_S = 0$, which implies $\omega^X(\mathfrak{t}_{\partial S}(P^a)) = 0$ and thus
\begin{equation}
\bra{ \Omega} \rho^X(\mathfrak{t}_{\partial S}(P^a)) \ket{\Omega}
=
\omega^X(\mathfrak{t}_{\partial S}(P^a))
=
0.
\end{equation}
Hence $\rho^X(\mathfrak{t}_{\partial S}(P^a)) \ket{\Omega} = 0$.

Now, since $u \in \fA(\Lambda)_{\loc}$ and $\rho^X$ is localized in $\Lambda^c$, we have $\rho^X(u) = u$. 
Furthermore, since $\supp(u)^{+1} \subseteq S$, we have that $u$ and $\mathfrak{t}_{\partial S}(P^a)$ are supported on disjoint regions. 
	Hence they commute and we get 
\begin{align}
\bra{\Omega} u \Dr_R(\varsigma_a) \ket{\Omega}
&= 
\bra{ \Omega}\rho^X(u) \Dr_R(\varsigma_a) \ket{\Omega}
\\&=
\bra{ \Omega} \rho^X(u)\rho^X( \mathfrak{t}_{\partial S}(P^a)) \Dr_R(\varsigma_a) \ket{\Omega}
\\&=
\bra{ \Omega} \rho^X(u \mathfrak{t}_{\partial S}(P^a)) \Dr_R(\varsigma_a) \ket{\Omega}
\\&=
\bra{ \Omega} \rho^X(\mathfrak{t}_{\partial S}(P^a)  u) \Dr_R(\varsigma_a) \ket{\Omega}
\\&=
\bra{\Omega} \rho^X(\mathfrak{t}_{\partial S}(P^a))  \rho^X(u) \Dr_R(\varsigma_a) \ket{\Omega}
\\&=
0,
\end{align}
which completes the proof.
\end{proof}

\begin{rem}
The paper \cite{2605.10693} provides a different physical perspective for Haag duality based in the local topological order axioms of \cite{MR4945955}. 
This paper provides an axiom termed LTO-HD, based on the approach in \cite{2509.23734}, that is a sufficient condition for Haag duality (when combined with a technical assumption). 
It is not hard to directly verify that LTO-HD fails to be satisfied for a nonabelian anyon state for the Levin--Wen model, but a rigorous treatment is beyond the scope of this paper. 
\end{rem}

\section{Failure of steering for nonabelian anyon state}

We now prove the failure of steering for nonabelian anyon states in the Levin--Wen model.
Recall from the introduction that for a bipartite system in a state $\omega$ with local, commuting von Neumann algebras $M_A$ and $M_B$, associated to parties $A$ and $B$, respectively, we say that $B$ can steer $A$ if for any normal functional $\psi$ on $M_A$ such that $0\leq \psi \leq \omega|_A$ there exists an operator $b\in M_B$ such that
\begin{align}
	\psi(a) = \omega(a b) 
\end{align}
for all $a\in M_A$.
In our case, we are interested in the nonabelian anyon state $\omega$ associated to a nonabelian anyon sector $\rho$.
As before we work in the vacuum representation, so that $\omega = \bra\Omega \cdot \ket\Omega$, but $M_A = \rho_\Lambda(\cR(\Lambda))$ and $M_B = \rho_{\Lambda^c}(\cR(\Lambda^c))$. 

\begin{thm}
\label{thm:LWFailureOfSteering}
	Let $\rho \colon \fA \to B(\cH)$ be a nonabelian anyon sector in the Levin--Wen model localized in the cone $\Lambda$. Set $A=\Lambda$ and $B = \Lambda^c$,  $M_A = \rho_\Lambda(\cR(\Lambda))$ and $M_B = \rho_{\Lambda^c}(\cR(\Lambda^c))$. 
 Then $B$ can steer $A$ relative to the state $\omega = \bra{\Omega} \cdot \ket{\Omega}$ if and only if $\rho$ is abelian.
\end{thm}

Recall that $B$ being able to steer $A$ steering is equivalent to the existence of a (normal) conditional expectation $\cE\colon M_B' \to M_A$ such that $\omega \circ \cE = \omega$ on $M_B'$ \cite[Thm.~B]{luijkQuantumSteeringEquivalent2026}. In our case, since the vacuum fulfills Haag duality and $\rho$ is localized in $\Lambda$, we have $\rho_{\Lambda^c}(\cR(\Lambda^c))' = \cR(\Lambda)$. Therefore we need a conditional expectation $\cE \colon \cR(\Lambda) \to \rho_\Lambda(\cR(\Lambda))$ such that $\bra{\Omega} x \ket{\Omega} = \bra{\Omega} \cE(x) \ket{\Omega}$ for all $x \in \cR(\Lambda)$.
By the following standard lemma, $\rho_\Lambda(\cR(\Lambda) \subseteq \rho_{\Lambda^c}(\cR(\Lambda^c))'$ is an irreducible subfactor, so there is at most one conditional expectation $\cE \colon \rho_{\Lambda^c}(\cR(\Lambda^c))' \to \rho_\Lambda(\cR(\Lambda))$ and it is faithful if it exists \cite[App. A]{kosakiLectureNotes}.\footnote{Faithfulness is immediate since the support projection $e$ of a conditional expectation $\cE \colon M\to N$ is in the relative commutant $N'\cap M$: $e$ is the smallest projection in $M$ such that $\cE(e) = 1$. But for any unitary $u \in N$ bimodularity implies $\cE(u e u^*) = uu^* = 1$. Hence $e$ has to commute with all unitaries in $N$.} 

\begin{lem}
\label{lem:AnyonsGiveIrrSubfactors}
Let $\rho \colon \fA \to B(\cH)$ be a nonabelian (irreducible) anyon sector localized in the cone $\Lambda$.
Then $\rho_\Lambda(\cR(\Lambda)) \subseteq \rho_{\Lambda^c}(\cR(\Lambda^c))'$ is an irreducible subfactor. 
\end{lem}

\begin{proof}
Since $\rho \colon \fA \to B(\cH)$ is an irreducible representation, we have that $\rho(\fA)'' = B(\cH)$, and thus we have 
\begin{equation}
\rho_\Lambda(\cR(\Lambda))' \cap \rho_{\Lambda^c}(\cR(\Lambda^c))'
=
(\rho_\Lambda(\cR(\Lambda)) \vee \rho_{\Lambda^c}(\cR(\Lambda^c)))'
=
\rho(\fA)'
=
\bbC.
\end{equation}
Hence $\rho_\Lambda(\cR(\Lambda)) \subseteq \rho_{\Lambda^c}(\cR(\Lambda^c))'$ is an irreducible subfactor by definition.
\end{proof}

In section \ref{sec:ConditionalExpectation}, we write down the unique conditional expectation when $\rho$ is an anyon sector with conjugate sector $\rho^\vee$. 
In section \ref{sec:LWFailureOfSteering}, we then show that this conditional expectation preserves $\bra{\Omega} \cdot \ket{\Omega}$ on $\cR(\Lambda)$ if and only if $\rho$ is abelian. 

\subsection{Conditional expectation}
\label{sec:ConditionalExpectation}

Let $\rho \colon \fA \to B(\cH)$ be a nonabelian anyon sector with conjugate sector $\rho^\vee \colon \fA \to B(\cH)$. 
Suppose $\rho$ and $\rho^\vee$ are localized in the cone $\Lambda$. 
Since $\rho^\vee$ is a conjugate sector, we have maps $R \colon \mathds{1} \to \rho^\vee \otimes \rho$ and $\bar R \colon \mathds{1} \to \rho \otimes \rho^\vee$ that satisfy the zig-zag equations \eqref{eq:zig-zag}.
Now, since the net $\cR(\Lambda)$ satisfies Haag duality, we have that $R, \bar R \in \cR(\Lambda)$. 
We have that $R \colon \mathds{1} \to (\rho^\vee)_\Lambda \otimes \rho_\Lambda$ and $\bar R \colon \mathds{1} \to \rho_\Lambda \otimes (\rho^\vee)_\Lambda$, so $(\rho^\vee)_\Lambda \in \End(\cR(\Lambda))$ is a dual object for $\rho_\Lambda \in \End(\cR(\Lambda))$. 
By \cite{MR1444286}, there exists a conditional expectation $\cE \colon \cR(\Lambda) \to \rho_\Lambda(\cR(\Lambda))$ given by 
\begin{equation}
\label{eq:ConditionalExpectation}
\cE(x)
\coloneqq
d_\rho^{-1} \rho_\Lambda(R^*(\rho^\vee)_\Lambda(x)R).
\end{equation}
Indeed, this is easily seen by noting that the range of $\cE$ is contained in $\rho(\cR(\Lambda))$ and the map $x\mapsto \phi_\rho(x) = d_\rho^{-1} R^*(\rho^\vee)_\Lambda(x) R$ is a left-inverse to $\rho_\Lambda$: 
Since $R\colon \mathds{1}\to \rho^\vee \otimes \rho$, we have
\begin{align}
(\phi_\rho \circ \rho_\Lambda)(x)
=
d_\rho^{-1} R^* (\rho^\vee)_\Lambda(\rho_\Lambda(x))R = d_\rho^{-1} x R^* R = x
\end{align}
for $x \in \cR(\Lambda)$. 
Note that this implies $\cE = \rho_\Lambda \circ \phi_\rho$.

In the case where $x \colon \rho \otimes \sigma \to \rho \otimes \tau$ for some sectors $\sigma, \tau$, we can depict $\cE(x)$ graphically as the following: 
\begin{equation}
\cE(x)
=
d_\rho^{-1}\,
\tikzmath{
\draw[mid>] (0, 0) -- (0, 1.5);
\draw[mid<] (.75, .75) -- (.75, .5) arc(0:-180:.25) -- (.25, 1) arc(180:0:.25) -- (.75, .75);
\draw[red] (1, 0) -- (1, .75);
\draw[orange] (1, .75) -- (1, 1.5);
\roundNbox{fill=white}{(.875, .75)}{.3}{0}{0}{$x$}
\node at (0, -.15) {$\scriptstyle \rho$};
\node[red] at (1, -.15) {$\scriptstyle \sigma$};
\node[orange] at (1, 1.65) {$\scriptstyle \tau$};
}
\end{equation}

We now make the following observation, which will be useful later. 

\begin{lem}
\label{lem:JonesProjectionRelation}
Let $P \coloneqq d_\rho^{-1} \bar R \bar R^* \colon \rho \otimes \rho^\vee \to \rho \otimes \rho^\vee$. 
Then $P$ is a projection in $\cR(\Lambda)$ satisfying that $\cE(P) = d_\rho^{-2} \mathds{1}$. 
\end{lem}

\begin{proof}
We first observe that $P$ is a projection since $\bar R^* \bar R = d_\rho I$. 
Furthermore, by Haag duality, $P \in \cR(\Lambda)$.
Graphically, we can represent $P$ as 
\begin{equation}
\label{eq:JonesProjection-Graphical}
P
=
d_\rho^{-1} \,
\tikzmath{
\draw[mid>] (0, 0) arc(180:0:.25);
\draw[mid<] (0, .75) arc(-180:0:.25);
}.
\end{equation}
Using the graphical calculus, we obtain that 
\begin{equation}
\cE(P)
=
d_\rho^{-2} \, 
\tikzmath{
\draw[mid>] (0, 0) -- (0, 1.75);
\draw[mid<] (1.25, 0) -- (1.25, .5) arc(0:180:.25) arc(0:-180:.25) -- (.25, 1.25) arc(180:0:.25) arc(-180:0:.25) -- (1.25, 1.75);
}
=
d_\rho^{-2} \, 
\tikzmath{
\draw[mid>] (0, 0) -- (0, 1.75);
\draw[mid<] (.5, 0) -- (.5, 1.75);
},
\end{equation}
which is exactly the desired equation.
\end{proof}

	As a consequence of the proof in Theorem~\ref{thm:dhr}, $\cE$ is non-trivial if and only if $\rho$ is nonabelian:	
\begin{lem}\label{lem:projection} The following are equivalent:
		\begin{enumerate}
				\item $\cE(P) = P$,
				\item $\rho$ is abelian ($d_\rho = 1$),
				\item $P = \mathds{1}$,
				\item $\cE = \id$.
		\end{enumerate}
\end{lem}
\begin{proof}
	Since $\cE(P) = d_\rho^{-2}\mathds{1}$, we clearly have $\cE(P) = P$ if and only if $P = \mathds{1}$.
	Next we show that $P=\mathds{1}$ if and only if $d_\rho = 1$, i.e., $\rho$ is abelian: First suppose $d_\rho = 1$ and assume that $P = \mathds{1}- Q$ for a projection $Q$. Since $\cE(P) = \mathds{1}$, it follows that $\cE(Q) = 0$. 
	Since the inclusion $\rho_\Lambda(\cR(\Lambda))\subset \cR(\Lambda)$ is irreducible, $\cE$ is faithful and it follows that $Q=0$, hence $P = \mathds{1}$. 
	Conversely, suppose $P= \mathds{1}$. Then $\mathds{1} = \cE(\mathds{1}) = \cE(P) = d_\rho^{-1}\mathds{1}$, hence $d_\rho = 1$.
	Finally, $\rho_\Lambda(\cR(\Lambda)) = \cR(\Lambda)$ (i.e. $\cE=\id$) if and only if $\rho_\Lambda$ is an automorphism on $\cR(\Lambda)$, which is the case if and only if $\rho$ is abelian by the proof of Theorem~\ref{thm:dhr}. 
\end{proof}

\subsection{Failure of steering for Levin--Wen nonabelian anyons}
\label{sec:LWFailureOfSteering}
We now recall the setup for the Levin--Wen models described in section \ref{sec:LWRigorousHeuristic}. 
Let $X \in \Irr Z(\cC)$ be a nonabelian anyon type, where $\cC$ is a unitary fusion category.  
Let $\scrC \coloneqq (L_n)_{n \geq 1}$ be a chain contained in the cone $\Lambda$. 
Recall that $\rho^X_{\scrC} \colon \fA \to \fA$ defines a nonabelian anyon sector of type $X$ localized in $\Lambda$. 
We then have that $\rho^{\bar X}_{\scrC} \colon \fA \to \fA$ is a conjugate sector for $\rho^X_{\scrC}$ \cite[Thm.~1.1]{2603.01936}. 
For ease of notation, we will write $\rho^X$ and $\rho^{\bar X}$ to denote $\rho^X_{\scrC}$ and $\rho^{\bar X}_{\scrC}$ respectively. 
We let $\cE \colon \cR(\Lambda) \to \rho^X_\Lambda(\cR(\Lambda))$ be the conditional expectation defined in section \ref{sec:ConditionalExpectation}. 
We show that $\bra{\Omega} \cdot \ket{\Omega} \neq \bra{\Omega} \cE(\cdot) \ket{\Omega}$ if and only if $\rho^X$ is nonabelian.
To do so, we show that $P\ket{\Omega} = \ket{\Omega}$, where $P \coloneqq d_X^{-1} \bar R \bar R^*$ is the projection depicted in \eqref{eq:JonesProjection-Graphical}. 
Lemma~\ref{lem:projection} then implies the result.

We first write down an explicit formula for $\bar R \colon \mathds{1} \to \rho^X \otimes \rho^{\bar X}$, from which we can build the projection $P = d_X^{-1} \bar R \bar R^*$. 
Recall the definition of the unitary $u^X_L$ for a link $L$ (section \ref{subsec:DrinfeldInsertions}), and recall that $U^X_{\scrC[n \to 1]} = u^X_{L_n} \cdots u^X_{L_1}$ (section \ref{subsec:LWAnyonSectors}). 
For ease of notation, we write $U_n^X \coloneqq U^X_{\scrC[n \to 1]}$ and $U_n^{\bar X} \coloneqq U^{\bar X}_{\scrC[n \to 1]}$.

\begin{lem}
\label{lem:ConjugateMaps-LevinWen}
For $n \geq 1$, define the maps 
\begin{equation}
\label{eq:ConjugateMaps-LevinWen}
R_n
\coloneqq
( U_{n + 1}^{\bar X})^*(U_n^{X})^*
\Dr_{L_{n + 1}} \left[ 
\tikzmath{
\draw[mid>, very thick, red] (0, .75) arc(-180:0:.25);
\draw[thick, dotted] (0, 0) -- (0, .25);
\draw[thick, dotted] (.5, 0) -- (.5, .25);
\node[red] at (0, .9) {$\scriptstyle \bar X$};
\node[red] at (.5, .9) {$\scriptstyle X$};
}\right],
\qquad\qquad
\bar R_n
\coloneqq
( U_{n + 1}^X)^*(U_n^{\bar X})^*
\Dr_{L_{n + 1}} \left[ 
\tikzmath{
\draw[mid<, very thick, red] (0, .75) arc(-180:0:.25);
\draw[thick, dotted] (0, 0) -- (0, .25);
\draw[thick, dotted] (.5, 0) -- (.5, .25);
\node[red] at (0, .9) {$\scriptstyle X$};
\node[red] at (.5, .9) {$\scriptstyle \bar X$};
}\right]
\end{equation}
The sequences $(R_n)$ and $(\bar R_n)$ converge in SOT to morphisms $R \colon \mathds{1} \to \rho^{\bar X} \otimes \rho^X$ and $\bar R \colon \mathds{1} \to \rho^X \otimes \rho^{\bar X}$ that satisfy the zig-zag equations. 
\end{lem}

We provide further clarification about the operators defined in \eqref{eq:ConjugateMaps-LevinWen}. 
Recall that the Drinfeld insertion operators \eqref{eq:DrinfeldInsertions} require a tensor factor corresponding to each puncture in both the source and the target. 
We therefore use the dotted lines to indicate that the object $\mathds{1}$ is used for each tensor factor in the source. 
Furthermore, we are suppressing the (non-canonical) fixed unitary $X^\vee \to \bar X$ in equation \eqref{eq:ConjugateMaps-LevinWen}. 

\begin{proof}[Proof of Lemma \ref{lem:ConjugateMaps-LevinWen}]
The fact that $(R_n)$ and $(\bar R_n)$ converge in SOT to morphisms $R \colon \mathds{1} \to \rho^{\bar X} \otimes \rho^X$ and $\bar R \colon \mathds{1} \to \rho^X \otimes \rho^{\bar X}$ follows from \cite[Lem.~5.7]{2603.01936}. 
It follows that these morphisms satisfy the zig-zag equations by the proof of \cite[Thm.~1.1]{2603.01936}.
\end{proof}

\begin{lem}
\label{lem:JonesProjectionFixesState}
Let $P \coloneqq d_X^{-1} \bar R \bar R^*$, where $\bar R$ is the morphism defined in Lemma \ref{lem:ConjugateMaps-LevinWen}. 
Then $P \ket{\Omega} = \ket{\Omega}$. 
\end{lem}

\begin{proof}
	First, observe that since $(\bar R_n)$ is a bounded sequence, we have that 
\begin{align}
P
&=
d_X^{-1} \bar R \bar R^*
=
d_X^{-1} \lim \bar R_n \bar R_n^*
\\&=
d_X^{-1} \lim ( U_{n + 1}^X)^*(U_n^{\bar X})^*
\Dr_{L_{n + 1}} \left[ 
\tikzmath{
\draw[mid>, very thick, red] (0, 0) arc(180:0:.25);
\draw[mid<, very thick, red] (0, .75) arc(-180:0:.25);
}\right]
U_n^{\bar X} U_{n + 1}^X
\\&=
\lim ( U_{n + 1}^X)^*(U_n^{\bar X})^* P_{L_{n + 1}}^{X \bar X} U_n^{\bar X} U_{n + 1}^X,
\end{align}
where in the last line we used the formula for $P^{X \bar X}_{L_{n + 1}}$ given in \eqref{eq:ProjectionForAnyonsAnnihilating} and all limits are in SOT.

We now compute $P \ket{\Omega}$. 
For $m \leq n$, we let $L_{n \to m}$ be the link formed by concatenating the links $L_n, \dots, L_m$. 
We also let $e_n \coloneqq \partial_i L_n$ and $e_{n - 1} \coloneqq \partial_f L_n$. 
(Recall that the punctures at the endpoints of links occur at an edge and the adjoining two faces.)
Note that $\ket{\Omega} \in P^{\mathds1 \mathds1}_{L_{n+1 \to 1}}\cH$, so by \cite[Lem.~6.9]{2511.21521} we have that 
\begin{equation}
\label{eq:FirstStepOfProjectionAppliedToGroundState}
U_{n + 1}^X\ket{\Omega} 
= 
u_{L_{n + 1 \to 1}}^X \ket{\Omega}
=
\Dr_{L_{n + 1 \to 1}}^{(\mathds 1 \mathds 1 \to X \bar X)}\ket{\Omega}.
\end{equation}
Therefore, by \cite[Lem.~6.8]{2511.21521}, we have that
\begin{equation}
\label{eq:FirstStepOfProjectionAppliedToGroundState-Constraint}
\mathfrak{t}_{e_{n + 1}}(p^X)U_{n + 1}^X\ket{\Omega} 
= 
\mathfrak{t}_{e_0}(p^{\bar X})U_{n + 1}^X\ket{\Omega} 
=
B_fU_{n + 1}^X \ket{\Omega}
=
U_{n + 1}^X\ket{\Omega},
\qquad\qquad 
\forall f \ \text{s.t.} \ e_0, e_{n + 1} \notin f.
\end{equation}
Graphically, we can represent $U_{n + 1}^X\ket{\Omega}$ by
\begin{equation}
d_X^{1/2}
\tikzmath{
\filldraw[gray!25](-.5, -.5) rectangle (6, 1);
\filldraw[fill=white] (0, 0) rectangle (1.25, .5);
\draw[dotted] (.625, 0) -- (.625, .5);
\node at (.79, .25) {\tiny $e_0$};
\draw[mid>, very thick, red] (1.25, .25) -- (4.25, .25);
\filldraw[fill=white] (4.25, 0) rectangle (5.5, .5);
\draw[dotted] (4.875, 0) -- (4.875, .5);
\node at (5.2, .25) {\tiny $ e_{n + 1}$};
\node[red] at (2.75, .5) {$\scriptstyle X$};
}
\end{equation}
As before, this picture should be thought of as a useful bookkeeping device, as the pictures only are defined mathematically for the finite volume Hilbert spaces. 

Now, by \eqref{eq:FirstStepOfProjectionAppliedToGroundState-Constraint}, we have that $U_{n + 1}^X\ket{\Omega} \in P_{L_{n \to 1}}^{\mathds{1} X} \cH$, so applying \cite[Lem.~6.9]{2511.21521} again, we have that 
\begin{equation}
U_n^{\bar X} U_{n + 1}^X\ket{\Omega}
=
u_{L_{n \to 1}}^{\bar X} U_{n + 1}^X\ket{\Omega}
=
\Dr_{L_{n \to 1}}^{(\mathds{1} \bar X \to \bar X \mathds{1})} U_{n + 1}^X\ket{\Omega}
=
\Dr_{L_{n \to 1}}^{(\mathds{1} \bar X \to \bar X \mathds{1})} \Dr_{L_{n + 1 \to 1}}^{(\mathds 1 \mathds 1 \to X \bar X)}\ket{\Omega},
\end{equation}
where we apply \eqref{eq:FirstStepOfProjectionAppliedToGroundState} in the last step. 
By the proof of \cite[Lem.~6.9]{2511.21521}, we can see that 
\begin{equation}
U_n^{\bar X} U_{n + 1}^X\ket{\Omega}
=
\Dr_{L_{n \to 1}}^{(\mathds{1} \bar X \to \bar X \mathds{1})} \Dr_{L_{n + 1 \to 1}}^{(\mathds 1 \mathds 1 \to X \bar X)}\ket{\Omega}
=
\Dr_{L_{n + 1 \to n}}^{(\mathds 1 \mathds 1 \to X \bar X)} \ket{\Omega}
\in
P_{L_{n + 1}}^{X \bar X} \cH,
\end{equation}
where in the last step we apply \cite[Lem.~6.8]{2511.21521}.
Graphically, the state $U_n^{\bar X}U_{n + 1}^X\ket{\Omega}$ is represented by 
\begin{equation}
d_X^{1/2}
\tikzmath{
\filldraw[gray!25](-.5, -.5) rectangle (6, 1);
\filldraw[fill=white] (2.5, 0) rectangle (3.75, .5);
\draw[dotted] (3.125, 0) -- (3.125, .5);
\node at (3.29, .25) {\tiny $e_n$};
\draw[mid>, very thick, red] (3.75, .25) -- (4.25, .25);
\filldraw[fill=white] (4.25, 0) rectangle (5.5, .5);
\draw[dotted] (4.875, 0) -- (4.875, .5);
\node at (5.2, .25) {\tiny $ e_{n + 1}$};
\node[red] at (4, .5) {$\scriptstyle X$};
}
\end{equation}
We thus have that $P_{L_{n + 1}}^{X \bar X} U_n^{\bar X}U_{n + 1}^X\ket{\Omega} = U_n^{\bar X}U_{n + 1}^X\ket{\Omega}$, so we have that
\begin{equation}
P\ket{\Omega}
=
\lim ( U_{n + 1}^X)^*(U_n^{\bar X})^* P_{L_{n + 1}}^{X \bar X} U_n^{\bar X} U_{n + 1}^X \ket{\Omega}
=
\lim ( U_{n + 1}^X)^*(U_n^{\bar X})^* U_n^{\bar X} U_{n + 1}^X \ket{\Omega}
=
\ket{\Omega},
\end{equation}
as desired.
\end{proof}

\begin{proof}[Proof of \cref{thm:LWFailureOfSteering}]
	First, since $\cE = \id$ clearly preserves $\omega$, steering is possible if $\rho$ is abelian by \cref{lem:projection}. So assume $d_X > 1$ from now.
By Lemma \ref{lem:JonesProjectionFixesState}, we have that $\bra{\Omega} P \ket{\Omega} = 1$. 
On the other hand, by Lemma \ref{lem:JonesProjectionRelation}, we have that 
\begin{equation}
\bra{\Omega} \cE(P) \ket{\Omega}
=
d_X^{-2} \langle \Omega| \Omega \rangle
=
d_X^{-2}.
\end{equation}
	We thus find $\bra{\Omega} P \ket{\Omega} \neq \bra{\Omega} \cE(P) \ket{\Omega}$.
	Hence $\cE$ does not preserve the sate $\omega = \bra{\Omega}\cdot\ket{\Omega}$, showing that steering is impossible.
\end{proof}

\section{Failure of approximate Haag duality}
\label{sec:approximateHD}
We now show that the weaker condition of approximate Haag duality fails for nonabelian anyon states. 
To define approximate Haag duality, we must introduce some notation. 
First, for a cone $\Lambda$, we let $\arg(\Lambda) \in (0, 2\pi)$ to be the angle between its bounding rays: 
\begin{equation}
\tikzmath{
\draw[red] (-.5, .866) edge [bend left=30] (.5, .866);
\node[red] at (0, 1.15) {$\scriptstyle \varphi$};
\draw[thick] (-1, 1.732) -- (0, 0) -- (1, 1.732);
\node at (-.4, 1.5) {$\scriptstyle \Lambda$};
}
\qquad\qquad 
\arg(\Lambda) = \varphi.
\end{equation}
Next, for $\varepsilon \geq 0$ and $R \geq 0$, we let $\Lambda_{R, \varepsilon}$ be the cone obtained by translating $\Lambda$ backwards along its axis by $R$ and widening it by $\varepsilon$: 
\begin{equation}
\Lambda_{R, \varepsilon}
=
\tikzmath{
\draw[orange, dashed] (0, 0) -- (0, -.866);
\node[orange] at (.15, -.25) {$\scriptstyle R$};
\draw[red] (1.05, .9526) edge[bend left=20] (1.25, 0.75775);
\node[red] at (1.23, .97) {$\scriptstyle \varepsilon$};
\draw[thick, gray] (-1, 1.732) -- (0, 0) -- (1, 1.732);
\node[gray] at (0, 1.5) {$\scriptstyle \Lambda$};
\draw[thick, dashed, gray] (-1.5, 1.732) -- (0, -.866) -- (1.5, 1.732);
\draw[thick] (-2, 1.732) -- (0, -.866) -- (2, 1.732);
}
\end{equation}

\begin{defn}[{\cite{MR4362722, 2511.08382}}]
Let $\cR(\Lambda) \subseteq B(\cH)$ be the net of cone von Neumann algebras coming from a quantum spin system (see \cref{sec:OABasics}). 
We say that the net $\cR(\Lambda)$ satisfies \emph{approximate Haag duality} if the following holds: 
For any $\varphi \in (0, 2\pi)$ and $\varepsilon > 0$ satisfying $\varphi + 4\varepsilon < 2\pi$, there exists $R > 0$ and decreasing functions $f_\delta(t)$ with $\lim_{t \to \infty} f_\delta(t) = 0$ such that
\begin{enumerate}
\item
for any cone $\Lambda$ with $\arg(\Lambda) = \varphi$, there is a unitary $U \in B(\cH)$ such that 
\begin{equation}
U\cR(\Lambda^c)'U^*
\subseteq
\cR(\Lambda_{R, \varepsilon}),
\end{equation}
\item
for any $\delta > 0$ and $t \geq 0$, there exists a unitary $\widetilde U \in \cR(\Lambda_{t, \varepsilon + \delta})$ such that 
\begin{equation}
\|U - \widetilde{U}\| \leq f_\delta(t).
\end{equation}
\end{enumerate}
\end{defn}

\begin{thm}
\label{thm:ApproxHaagDualityFails}
Suppose the net of reference von Neumann algebras $\cR(\Lambda)$ satisfies Haag duality. 
Let $\rho \colon \fA \to B(\cH)$ be a nonabelian anyon sector localized in a cone $\Lambda$ with conjugate sector $\rho^\vee \colon \fA \to B(\cH)$ (also localized in $\Lambda$). 
Then the net of von Neumann algebras $\rho_\Delta(\cR(\Delta))$ does not satisfy approximate Haag duality. 
\end{thm}

\begin{proof}
Observe that for any cone $\Delta$ such that $\Lambda \subseteq \Delta$, we have that $\rho$ and $\rho^\vee$ are localized in $\Delta$. 
Let $\cE_\Delta \colon \cR(\Delta) \to \rho_\Delta(\cR(\Delta))$ be the conditional expectation given in \eqref{eq:ConditionalExpectation}. 
Let $R \colon \mathds{1} \to \rho^\vee \otimes \rho$ and $\bar R \colon \mathds{1} \to \rho \otimes \rho^\vee$ be morphisms satisfying the zig-zag equations \eqref{eq:zig-zag}. 
We let $P \coloneqq d_\rho^{-1} \bar R \bar R^* \colon \rho \otimes \rho^\vee \to \rho \otimes \rho^\vee$. 
	By Lemma \ref{lem:JonesProjectionRelation}, $P$ is a projection in $\cR(\Lambda)$ satisfying that $\cE_\Lambda(P) = d_\rho^{-2} \mathds{1}$. Since $\rho_\Delta(x) = \rho_\Lambda(x)$ for any $x\in \cR(\Lambda)\subset \cR(\Delta)$, it follows from the construction of $\cE_\Delta$ that $\cE_\Delta(x) = \cE_\Lambda(x)$ for any $x\in \cR(\Lambda)$. Hence $\cE_\Delta(P) = d_\rho^{-2} \mathds{1}$, too.
	Since $\rho$ represents a nonabelian anyon we have $P \neq d_\rho^{-2} \mathds{1}$. 
We set $\gamma \coloneqq \|P - d_\rho^{-2} \mathds{1}\| > 0$.

Suppose, towards contradiction, that the net $\rho_\Lambda(\cR(\Lambda))$ does satisfy approximate Haag duality. 
Then there exist $R, \varepsilon > 0$ and $U \in B(\cH)$ such that 
\begin{equation}
U \cR(\Lambda)U^* 
= 
U\rho_{\Lambda^c}(\cR(\Lambda^c))'U^*
\subseteq
\rho_{\Lambda_{R, \varepsilon}} (\cR(\Lambda_{R, \varepsilon})),
\end{equation}
	where we used that $\rho$ is localized in $\Lambda$ together with Haag duality for the net $\cR(\Lambda)$ for the first equality.
	Moreover, there exist $t, \delta > 0$ and $\widetilde U \in \rho_{\Lambda_{t, \varepsilon + \delta}}(\cR(\Lambda_{t, \varepsilon + \delta}))$ such that 
\begin{equation}
\|U - \widetilde{U}\| < \frac{\gamma}{2}.
\end{equation}
Observe that 
\begin{equation}
\|UPU^* - \widetilde U P \widetilde U^*\|
\leq
\|UPU^* - \widetilde U P U^*\| + \|\widetilde UPU^* - \widetilde U P \widetilde U^*\|
\leq
	2\|U - \widetilde U\|
<
\gamma. 
\end{equation}

Now, let $\Delta$ be a cone such that $\Lambda_{R, \varepsilon}, \Lambda_{t, \varepsilon + \delta} \subseteq \Delta$. 
(For instance, one can take $\Delta = \Lambda_{R + t, \varepsilon + \delta}$.) 
Then we have that $U\cR(\Lambda)U^* \subseteq \rho_\Delta(\cR(\Delta))$ and $\widetilde{U} \in \rho_\Delta(\cR(\Delta))$. 
	In particular, this implies $\cE_\Delta(UPU^*) = UPU^*$. Using the bimodule property of $\cE_\Delta$ we therefore have that 
\begin{equation}
\|\cE_\Delta(UPU^*) - \cE_\Delta(\widetilde U P \widetilde U^*)\|
=
\|UPU^* - \widetilde U\cE_\Delta(P) \widetilde U^*\|
=
\|UPU^* - d_\rho^{-2} \mathds{1}\|
=
\|P - d_\rho^{-2} \mathds{1}\|
=
\gamma.
\end{equation}
On the other hand, since $\|UPU^* - \widetilde U P \widetilde U^*\| < \gamma$, we have that 
\begin{equation}
\|\cE_\Delta(UPU^*) - \cE_\Delta(\widetilde U P \widetilde U^*)\|
=
\|\cE_\Delta(UPU^* - \widetilde U P \widetilde U^*)\|
\leq
\|UPU^* - \widetilde U P \widetilde U^*\|
<
\gamma,
\end{equation}
which is a contradiction. 
Hence the net $\rho_\Delta(\cR(\Delta))$ does not satisfy approximate Haag duality. 
\end{proof}

We now briefly review the definitions of finite-range commuting projector Hamiltonians and gapped ground states (see \cite{MR3617688} for a more detailed discussion on this subject). 
We consider a Hamiltonian $H$ to be a collection of local Hamiltonians $H_R$ for $R \subseteq \Gamma$ finite satisfying that 
\begin{equation}
H_R
=
\sum_{X \subseteq R} \Phi(X),
\end{equation}
where each $\Phi(X) \in \fA(X)$ is self-adjoint. 
We term $\Phi(X)$ the \emph{interactions} of the Hamiltonian $H$.
We say that the Hamiltonian $H$ is \emph{finite-range} if there exists $r > 0$ such that $\Phi(X) = 0$ for all $X \subseteq \Gamma$ finite with $\diam(X) > r$. 
Furthermore, we say that $H$ is \emph{commuting projector} if each $\Phi(X)$ is a projection and $[\Phi(X), \Phi(Y)] = 0$ for all $X, Y \subseteq \Gamma$ finite.

For a Hamiltonian $H$ with finite-range interactions $\Phi(X)$, we define a derivation $\delta$ on $\fA_{\loc}$ by 
\begin{equation}
\delta(A)
\coloneqq
\lim_{\Lambda \nearrow \Gamma} i [H_\Lambda, A]
=
\sum_{\substack{X \subseteq \Gamma \ \text{finite} \\ A \cap \supp(\Phi(X)) \neq \emptyset}} i[\Phi(X), A].
\end{equation}
We say that a state $\omega \colon \fA \to \bbC$ is a \emph{gapped ground state} for $H$ if there exists $\gamma > 0$ such that
\begin{equation}
-i\omega(A^*\delta(A)) \geq \gamma \omega(A^*A), 
\qquad\qquad
\forall A \in \fA_{\loc}, \ \omega(A) = 0.
\end{equation}
If $H$ is a commuting projector Hamiltonian, we say that $\omega \colon \fA \to \bbC$ is a \emph{frustration-free ground state} for $H$ if $\omega(\Phi(X)) = 0$ for all $X \subseteq \Gamma$ finite. 
If $H$ has a unique frustration-free ground state, then this state is also a gapped ground state for $H$. 

\begin{cor}
There exists a gapped ground state of a finite-range commuting projector Hamiltonian that does not satisfy approximate Haag duality. 
\end{cor}

\begin{proof}
Recall the setup for the Levin--Wen model in \cref{sec:LWRigorousHeuristic}. 
We let $X \in Z(\cC)$ be a nonabelian anyon type, where $\cC$ is a unitary fusion category. 
Let $\scrC \coloneqq (L_n)_{n \geq 1}$ be a chain contained in the cone $\Lambda$, and let $e_0 \coloneqq \partial_f L_1$ be the edge at the endpoint puncture of this chain. 
We also let $f^1_{e_0}, f^2_{e_0}$ denote the two faces adjacent to $e_0$. 
Then $\rho^X \coloneqq \rho^X_{\scrC} \colon \fA \to \fA$ defines a nonabelian anyon of type $X$ localized in the cone $\Lambda$. 
As before, we let $\pi_0 \colon \fA \to B(\cH)$ be the GNS representation of the frustration-free ground state $\omega_0$ for the Levin--Wen model, and we let $\cR(\Delta) = \pi_0(\fA(\Delta))''$. 
Recall that $\pi_0 \circ \rho^X$ is a GNS representation for the state $\omega^X \coloneqq \omega_0 \circ \rho^X$ \cite[Prop.~7.2]{2511.21521}. 
By Theorem \ref{thm:ApproxHaagDualityFails}, the net $\rho^X_\Delta(\cR(\Delta))$ fails to satisfy approximate Haag duality. 
We now show that $\omega^X$ is a gapped ground state for some finite-range commuting projector Hamiltonian. 
Indeed, by \cite[Prop.~5.2]{2511.21521}, we have that $\omega^X$ is the unique frustration-free ground state of the commuting projector Hamiltonian defined by the following interactions: 
\begin{equation}
\Phi(e) = \mathds{1} - A_e \ \forall e \neq e_0,
\qquad
\Phi(f) = \mathds{1} - B_f \ \forall f \neq f^1_{e_0}, f^2_{e_0},
\qquad
\Phi(f^1_{e_0} \cup f^2_{e_0}) = \mathds{1} - \mathfrak{t}_{e_0}(p^{\bar X}).
\end{equation}
Therefore, $\omega^X$ is a gapped ground state for this Hamiltonian.
\end{proof}

\begin{rem}\label{rem:anyons}
Even though a nonabelian anyon sector $\rho \colon \fA \to B(\cH)$ fails to satisfy approximate Haag duality, there is still a way to define the braided monoidal category of superselection sectors for the gapped phase associated to $\rho$. 
We let $\pi_0 \colon \fA \to B(\cH)$ denote the reference representation satisfying Haag duality, with respect to which $\rho$ is a nonabelian anyon sector. 
Then $\pi \colon \fA \to B(\cH_1)$ is a superselection sector with respect to $\pi_0$ if and only if it is a superselection sector with respect to $\rho$. 
Indeed, since $\rho$ is an anyon sector with respect to $\pi_0$, we have that for every cone $\Lambda$, 
\begin{equation}
\rho|_{\fA(\Lambda^c)}
\cong
\pi_0|_{\fA(\Lambda^c)}.
\end{equation}
Therefore, for any cone $\Lambda$, we have that 
\begin{equation}
\pi|_{\fA(\Lambda^c)} \cong \rho|_{\fA(\Lambda^c)}
\qquad
\Longleftrightarrow
\qquad
\pi|_{\fA(\Lambda^c)} \cong \pi_0|_{\fA(\Lambda^c)},
\end{equation}
so $\pi$ is a superselection sector with respect to $\pi_0$ if and only if it is one with respect to $\rho$. 

We can therefore define the braided monoidal structure on the category of superselection sectors with respect to $\rho$ to be the one defined using $\pi_0$ as the reference representation. 
It remains to check that this fusion and braiding is an invariant of the gapped phase. 
We say that two states are in the same phase if they are related by an automorphism $\alpha$ obtained via spectral flow by finite-range Hamiltonians (see \cite{MR2885611} for a precise definition). 
By \cite[Thm.~3.1]{MR4426734}, such an automorphism is \emph{approximately factorizable} in the sense of \cite[Def.~1.2]{MR4362722}. 
For any approximately factorizable automorphism $\alpha \colon \fA \to \fA$, we have that $\pi_0$ and $\pi_0 \circ \alpha$ have the same braided monoidal category of superselection sectors by \cite[Thm.~6.1]{MR4362722}, and the equivalence of categories is given by the map $\pi \mapsto \pi \circ \alpha$. 
In particular, $\rho \circ \alpha$ is an anyon sector with respect to $\pi_0 \circ \alpha$  by \cite[Thm.~4.7]{MR4426734}, so $\pi \colon \fA \to B(\cH_1)$ is a superselection sector with respect to $\rho \circ \alpha$ if and only if it is one with respect to $\pi_0 \circ \alpha$. 
We can therefore define the braided monoidal structure on superselection sectors with respect to $\rho \circ \alpha$ to be the one obtained using $\pi_0 \circ \alpha$ as the reference representation. 
By \cite[Thm.~6.1]{MR4362722}, this is exactly the same category as the one we had previously defined for $\rho$ by way of $\pi_0$. 

We can therefore see that the anyon theories for the gapped phases corresponding to $\pi_0$ and $\rho$ are exactly the same, even though these gapped phases must be distinct since one satisfies approximate Haag duality while the other does not. 
This observation therefore strengthens our argument that approximate Haag duality should be viewed as a selection criterion for the collection of gapped phases. 
\end{rem}

\printbibliography

\end{document}